\documentclass[draftcls,onecolumn,journal]{IEEEtran}

\usepackage{algorithm}
\usepackage{algorithmic}
\usepackage{amsmath}

\usepackage{amsfonts}
\usepackage{amssymb}

\usepackage{tikz}
\usetikzlibrary{shapes}
\usetikzlibrary{arrows}
\usetikzlibrary{decorations.pathreplacing}
\usetikzlibrary{decorations.markings}

\usepackage{subfigure}

\usepackage[hidelinks=true]{hyperref}
\usepackage{ifthen}

\usepackage{graphicx}
\graphicspath{{./fig/}{./fig-prelim/}{./pdfFig/}{./pdfFig/sim/}}
\DeclareGraphicsExtensions{.pdf,.jpeg,.png}

\usepackage{makecmds}

\newenvironment{scope}{}{}

\newcommand{\emath}[1]{\ensuremath{#1}}

\newcommand{\expect}{\emath{\mathbb{E}}}

\ifdefined\kylehackusingbeamer
\else
	\newtheorem{assump}{Assumption}
	\newtheorem{theorem}{Theorem}[section]
	\newtheorem{definition}[theorem]{Definition}
	
	\newtheorem{lemma}[theorem]{Lemma}
	\newtheorem{corollary}[theorem]{Corollary}
	\newtheorem{remark}[theorem]{Remark}
\fi
	\newtheorem{prop}[theorem]{Proposition}

\provideenvironment{example}[1][Example]{
\begin{trivlist}
\item[\hskip \labelsep {\bfseries #1}]
\it	
}{ \end{trivlist} }

\newcommand{\onehalf}{\emath{ \frac{1}{2} }}
\newcommand{\reals}{ \emath{\mathbb{R}} }

\newcommand{\vecnorm}[1]{\emath{\left\|#1\right\|}}

\newcommand{\picktag}{\emath{ {\rm P} }}
\newcommand{\delvtag}{\emath{ {\rm D} }}

\newcommand{\env}{\emath{\Omega}}

\newcommand{\arrivalrate}{\emath{\lambda}}
\newcommand{\utilization}{\emath{\varrho}}
\newcommand{\servicetime}{\emath{s}}
\newcommand{\numveh}{\emath{m}}

\newcommand{\den}{\varphi}

\newcommand{\workspace}{\emath{\mathcal{C}}}

\newcommand{\workcell}{\emath{C}}
\newcommand{\cellset}{ {\mathbb C} }

\newcommand{\cyclevar}{\mathcal L}

\newcommand{\distfunc}{}
\renewcommand{\distfunc}[1]{\vecnorm{#1}}

\newcommand{\pathvar}{{\mathcal P}}
\newcommand{\pathset}{{\mathbb P}}

\newcommand{\nodeset}{V}

\newcommand{\edgeset}{{\mathcal A}}

\newcommand{\krondelta}{\mathbb{I}}

\providecommand{\argholder}{{\, \cdot \,}}

\usepackage{mathrsfs}

\newcommand{\STANDALONE}{}

\usepackage{multirow}
\usepackage{caption}

\usetikzlibrary{arrows,shapes,snakes,automata,backgrounds,petri}

\begin{document}

\newtheorem{conjecture}[theorem]{Conjecture}


\newcommand{\nodevar}{u}
\newcommand{\altnodevar}{v}

\newcommand{\arcset}{A}
\newcommand{\arcvar}{a}
\renewcommand{\edgeset}{\arcset}


\newcommand{\pointvar}{{\bf p}}
\newcommand{\altpointvar}{{\bf q}}

\newcommand{\distance}{{\mathscr D}}

\makecommand{\workspace}{{\mathcal W}}		

\providecommand{\numserved}{S}


\newcommand{\sigfield}{{\mathcal F}}
\newcommand{\altsigfield}{{\mathcal G}}

\newcommand{\measvar}{\mu}
\newcommand{\newmeas}{\tilde\mu}

\renewcommand{\den}{\varphi}
\newcommand{\cumden}{\Phi}
\newcommand{\invcumden}{\Psi}

\newcommand{\denone}{\den^\srcTag}
\newcommand{\dentwo}{\den^\trgTag}


\newcommand{\Wass}{{W}}
\newcommand{\srcTag}{\sharp}
\newcommand{\trgTag}{\flat}

\newcommand{\measone}{\mu^\srcTag}
\newcommand{\meastwo}{\mu^\trgTag}
\newcommand{\newmeasone}{\tilde\mu^\srcTag}
\newcommand{\newmeastwo}{\tilde\mu^\trgTag}


\newcommand{\roadnet}{{\mathcal R}}

\newcommand{\roadverts}{{\bf V}}
\newcommand{\roadvert}{{\bf v}}

\newcommand{\roadset}{{\bf R}}
\newcommand{\roadvar}{r}

\newcommand{\roadlen}{L}

\newcommand{\coordset}{S}
\newcommand{\coordvar}{s}
\newcommand{\coordmap}{\roadnet}


\newcommand{\precTag}{\prec}
\newcommand{\succTag}{\succ}

\newcommand{\thresh}{k}

\newcommand{\roadTL}[1][\roadvar]{(#1,#1^-)}
\newcommand{\roadTR}[1][\roadvar]{(#1,#1^+)}
\newcommand{\roadFL}[1][\roadvar]{(#1^-,#1)}
\newcommand{\roadFR}[1][\roadvar]{(#1^+,#1)}

\newcommand{\roadconnleft}[1][\roadvar]{ {\rm tconn}_{#1} }
\newcommand{\roadconnright}[1][\roadvar]{ {\rm hconn}_{#1} }

\newcommand{\tverts}{{\nodeset^\srcTag}}
\newcommand{\bverts}{{\nodeset^\trgTag}}




\newcommand{\topcopy}[1]{{{#1}\trgTag}}
\newcommand{\bottomcopy}[1]{{{#1}\srcTag}}


\newcommand{\tthresh}{\thresh}
\newcommand{\bthresh}{\thresh}

\newcommand{\randpt}{{P}}
\newcommand{\altrandpt}{{Q}}
\newcommand{\pickden}{{\den_\picktag}}
\newcommand{\delvden}{{\den_\delvtag}}

\newcommand{\MyTitle}{An Explicit Formulation of the Earth Mover's Distance with Continuous Road Map Distances}
\author{Kyle~Treleaven, Emilio~Frazzoli
\thanks{Kyle Treleaven and Emilio Frazzoli are with the Laboratory for Information and Decision Systems, Department of Aeronautics and
Astronautics, Massachusetts Institute of Technology, Cambridge, MA 02139 {\tt\small \{ktreleav, frazzoli\}@mit.edu}.}%
\thanks{This research was supported in part by the Future Urban Mobility project of the Singapore-MIT Alliance for Research and Technology (SMART) Center, with funding from Singapore's National Research Foundation. }
}
\ifthenelse{ \isundefined{\STANDALONE} }
{
}{
\title{\MyTitle}
\maketitle
}

\begin{abstract}

The \emph{Earth mover's distance} (EMD) is a measure of distance between probability distributions which
is at the heart of
\emph{mass transportation theory}
.
Recent research 
has shown that the EMD plays a crucial role in studying
the potential impact of Demand-Responsive Transportation (DRT) and Mobility-on-Demand (MoD) systems, which
are growing paradigms for one-way vehicle sharing 
where people drive (or are driven by) shared vehicles from a point of origin to a point of destination.
While the ubiquitous physical transportation setting is the ``road network'',
characterized by systems of roads connected together by interchanges,
most analytical works about vehicle sharing represent distances between points in a plane
using the simple Euclidean metric.
Instead, we consider the EMD when the ground metric is taken from a class
of one-dimensional, continuous metric spaces,
reminiscent of road networks.
We produce an ``explicit'' formulation
of the Earth mover's distance given any finite road network $\roadnet$.
The result generalizes the EMD with a Euclidean $\reals^1$ ground metric, which
had remained one of the only known \emph{non-discrete} cases with an explicit formula.
Our formulation casts the EMD as the optimal value of
a finite-dimensional, real-valued optimization problem,
with a convex objective function and linear constraints.
In the special case that the input distributions have piece-wise uniform (constant) density,
the problem reduces to one whose objective function is convex \emph{quadratic}.
Both forms are amenable to modern mathematical programming techniques.

\end{abstract}

\section{Introduction}

The \emph{Earth mover's distance} (EMD) is a measure of distance between probability distributions---%
or measures, more generally---%
which is commonly encountered in mathematics and computer science.
In mathematics, it is generally referred to as the Rubenstein/Kantorovich/Wasserstein distance,
or simply \emph{Wasserstein distance}.
The metric is also the solution to the Monge-Kantorovich problem, which is at the heart of
\emph{mass transportation theory}~\cite{rachev1998mass,rachev1998mass2}.
A common informal interpretation of the EMD is that
if one treats two measures (say, $\measone$ and $\meastwo$) as two distinct ways of arranging some fluid/continuous commodity
(e.g., ``a pile of dirt'') in a spatial domain $\env$,
then the EMD is the minimum cost of transforming the arrangement described by $\measone$ into the arrangement described by $\meastwo$.
Such interpretation requires that the underlying domain be equipped with a ``ground metric''
$\distance : \env \times \env \to \reals_{\geq 0}$
by which the cost of transformations can be measured;
the notion is that relocating a unit of commodity
from a point $\pointvar \in \env$ to point a $\pointvar' \in \env$
incurs cost $\distance(\pointvar,\pointvar')$.
Formally, the EMD is defined,
given a complete and separable metric space $(\env,\distance)$ as
\begin{equation} \label{eq:EMDdef}
  \Wass(\measone,\meastwo)
  \doteq
  \inf_{ \gamma \in \Gamma(\measone,\meastwo) }
    \int_{ \env } \distance(\pointvar, \pointvar') \ d\gamma(\pointvar,\pointvar').
\end{equation}
The search space $\Gamma$ is the set of \emph{couplings} of $\measone$ and $\meastwo$,
i.e., the collection of all joint measures over $\env^2$
having marginals $\measone$ and $\meastwo$ on the first and second factors, respectively.
Generally speaking, $\Gamma$ is infinite-dimensional.
%

\subsection{Literature Review}

The work on EMD has developed, to a large extent, in two separate and independent tracks:
the discrete case deals largely with optimization over finite-dimensional polyhedra, and has been examined by computer scientists;
the continuous case (and a unifying theory) has remained the subject of the more mathematical/theoretical study called
mass transportation theory.
One of the most successful recent applications of the EMD has been
in image matching and retrieval~\cite{werman1985distance,rubner1997earth,cohen1999earth,rubner2000earth},
toward the development of fast computerized image databases.
The EMD obtains several advantages over previously-used metrics for comparing certain image data
represented using \emph{histograms} (i.e., distributions of finite support).
The metric has also been studied recently from an algorithmic perspective~\cite{indyk2007near,ling2007efficient,andoni2008earth,shirdhonkar2008approximate,andoni2009efficient,indyk2009sublinear}, because
classical algorithms to compute the EMD can be too slow to meet the requirements of large database systems.
Many such studies leverage special structure of a particular ground metric.
While most algorithmic studies of the EMD consider that the two distributions, or histograms, are known \emph{a priori},
a study in~\cite{indyk2009sublinear} considers optimal approximation algorithms in the case that the distributions are not known,
but the samples used to compute the histograms are obtained as a ``streaming input''.
The EMD has applications in other computer science domains as well,
e.g., alignment of two-dimensional surfaces~\cite{lipman2011conformal}.
In~\cite{treleaven2013},
the EMD with a Euclidean ground metric in $\reals^d$
has been shown to factor in predicting the length of the optimal solution to the \emph{Stacker Crane problem} (SCP),
a tour through many randomly generated transportation demands.
The SCP is a combinatorial optimization problem with applications in vehicle routing, and
the prediction is in a sense parallel to the Beardwood-Halton-Hammersley (BHH) theorem~\cite{beardwood1959shortest}
about the related \emph{Traveling Salesman problem}.
Similarly, the EMD has applications in the formal analysis of Mobility-on-Demand systems.
For example, \cite{pavone2011load} and~\cite{treleaven2013} present 
conditions to ensure the stability of two Mobility-on-Demand scenarios which 
can be expressed in terms of the EMD.

\subsection{Motivation}

When $\env$ is a finite set, then it is straightforward to compute the EMD, regardless of the metric $\distance$.
For example, the EMD can be computed by reducing it to a network flow problem~\cite{ahuja1993network}.
In this paper, we call such a formulation \emph{explicit}, in a sense that
we will define formally in Section~\ref{sec:problem statement}.
Unfortunately, if
the ground domain $\env$ is not finite, then
explicit formulations of the EMD are only known in a few special cases, although
it is usually straight-forward
to obtain a $1+O(\epsilon)$ approximation in polynomial time.
(If $\env$ is not finite, but both distributions have finite \emph{support}, then
$\env$ can be restricted to a finite set appropriately.)
The finite case has received by far the most attention in recent years,
as progress on the general problem has stagnated.
All the above works except~\cite{lipman2011conformal} and~\cite{treleaven2013}
consider the \emph{discrete} version of the EMD.
Indeed, the term ``Earth Mover's distance'' seems to have been coined in~\cite{rubner1998metric}
by researchers studying the discrete case,
so the assumption of discrete domains is often implicit to its usage.

%

One of the only known \emph{non-discrete} cases with an explicit formula
is if $\env = \reals$ and $\distance(x,y) = |x-y|$.
Then
\begin{equation} \label{eq:Wass_real_line}
  \Wass(\measone,\meastwo)
  = \int \left| F^\srcTag(y) - F^\trgTag(y) \right| \ dy,
\end{equation}
where $F^i$ denotes the distribution function (d.f.) of a measure $\measvar^i$,
i.e., $F^i(y) := \measvar^i( \{ Y \in \reals \ : \ Y \leq y \} )$.
(If $\measone$ and $\meastwo$ are \emph{probability} distributions,
then $F^\srcTag$ and $F^\trgTag$ are their respective cumulative density functions.)
R\"{u}schendorf discusses a few other ``explicit'' expressions in~\cite{Ruschendorf1985};
however, as far as we are aware, the state-of-the-art has not improved significantly since the 1980s. 

This paper is motivated largely by the work in~\cite{pavone2011load,treleaven2013},
about a vehicle ``rebalancing'' problem that appears to be fundamental to 
large scale one-to-one transportation problems.
We consider the EMD when the domain $\env$ is taken from a class of one-dimensional metric spaces
inspired by spatial road networks, and which generalizes $\reals^1$:
Their metrics are almost everywhere \emph{locally} like Euclidean $\reals^1$, but
they may have a more general, ``graph-like'' topology.
We call such spaces, simply, \emph{roadmaps}.
Formal treatments of road networks as \emph{continuous} metric spaces are
somewhat rare in literature.
\cite{okabe2012spatial} explores one similar yet distinct branch of geometrical study.

\subsection{Contributions}

The main contribution of this paper is an explicit formulation
of the Earth mover's distance (EMD) $\Wass(\measone,\meastwo)$
for any road network $\roadnet$.
The result generalizes the formulation of the EMD in Euclidean $\reals^1$, which
(i) is the most trivial kind of road network, and 
(ii) had remained one of the only EMDs in a continuous domain with an explicit formula.
%
%
We find that even given quite general distributions,
e.g., those admitting density functions,
our formulation casts the EMD as the optimal value of
a finite-dimensional, real-valued optimization problem with a convex objective function and linear constraints, which
is amenable to
convex programming techniques~\cite{Bertsekas1999}.
In the special case that the distributions $\measone$ and $\meastwo$ have piece-wise uniform (constant) density,
the problem reduces to one whose objective function is convex \emph{quadratic},
in number of variables linear in the number of \emph{pieces}.
One can solve such a problem efficiently using standard quadratic programming (QP) methods.

\subsection{Applications to Vehicle Sharing}


Mobility-on-Demand (MoD) is a growing paradigm for one-way vehicle sharing~\cite{Mitchell2010}, where
people drive (or are driven by) shared vehicles from a point of origin to a point of destination.
%
Recent research~\cite{treleaven2013,pavone2011load} has shown that the EMD plays a crucial role in studying
the potential impact of MoD systems,
e.g., in terms of the fleet sizes required to meet quality of service objectives.
However, while the ubiquitous physical setting
of a vast number of transportation problems
is a ``road network'' characterized by systems of roads connected together by interchanges,
all the mathematically rigorous studies that we are aware of represent the distance between points
in a planar workspace $\workspace$
using the simple Euclidean metric.
At small-to-medium scale (e.g., of the so-called \emph{Last Mile}),
the Euclidean distance can yield a poor approximation of roadmap distances.
The results of this paper can be used to address such limitations.

\subsection{Organization}

The rest of the paper is organized as follows.
First, we state formally the objectives of the paper in Section~\ref{sec:problem statement}.
We present the relevant background in Section~\ref{sec:background},
including basic definitions in graph theory and geometry,
and a preliminary treatment of network flow theory and properties of the EMD.
In Section~\ref{sec:roadnet geometry}, we introduce the class of roadmap metric spaces
which form the basis of our analysis;
they provide the roadmap distance ground metrics commonly associated with road networks.
In Section~\ref{sec:main result},
we present the main result of the paper, 
an explicit formulation of the EMD on road networks
as a finite-dimensional convex optimization problem.
In Section~\ref{sec:simulation} we present the results of a simulation experiment designed to
validate our result while demonstrating the role of the EMD
in characterizing the ``workload''
faced by a one-way vehicle sharing system.
%
In Section~\ref{sec:general purpose approx}, we provide a naive, general-purpose procedure
to compute an approximation of the EMD for any ground metric.
In Section~\ref{sec:roadnet approx}, we refine the procedure
using structural knowledge about road networks
to obtain a procedure which is simultaneously more efficient and more insightful.
(These approximations are integral components to a formal proof
of the correctness of our main result,
presented later in the paper.)
In Section~\ref{sec:algorithm analysis} we analyze the computational space and runtime complexity of the procedures of
Sections~\ref{sec:main result}, \ref{sec:general purpose approx}, and~\ref{sec:roadnet approx}.
We provide the formal proof of correctness of our main result in Section~\ref{sec:proof of correctness}.
Finally, we present concluding remarks in Section~\ref{sec:conclusion}.

\section{Problem Statement}
\label{sec:problem statement}

For the rest of the paper,
we will say that a formula is \emph{explicit} if
it is a closed-form expression or an integral involving closed-forms, or
if it is a convex program in terms of such expressions for which strong duality holds~\cite[Ch.~5]{Boyd:2004}.
It is essentially straightforward to compute such formulas, because
closed-forms are ``well-studied'', and
efficient techniques exist both for numerical integration and convex optimization~\cite[Ch.~11]{Boyd:2004}.
Many of the distributions on $\reals$ which are commonly used
to represent other ones
have cdfs which are considered closed-form.
Network optimization problems~\cite[Ch.~5]{Bertsekas1999} are among a broad class of convex optimization problems satisfying strong duality.
%

The objective of the paper is
to obtain an explicit formulation of the Earth Mover's distance,
given a roadmap $\roadnet$,
as a network optimization problem.

\section{Background}
\label{sec:background}

\subsection{Notation}

\subsubsection{Graphs}
We use the following graph notation throughout the paper:
Let $(\nodeset,\arcset)$ denote a directed graph, or \emph{di-graph},
with vertex set $\nodeset$ and a set of directed edges $\arcset$.
In general, $(\nodeset,\arcset)$ might be a \emph{multi-} di-graph,
meaning that multiple distinct edges may share the same endpoints.
For any edge $\arcvar \in \edgeset$,
let $\arcvar^-$ denote the \emph{tail} of $\arcvar$
and let $\arcvar^+$ denote the \emph{head} of $\arcvar$.
For example, if $\arcvar=(u,v)$, then $\arcvar^- = u$ and $\arcvar^+ = v$.

\subsubsection{Geometry}
\providecommand{\pointset}{\env}

\begin{definition}[Metric space]
A \emph{metric space} is the pair of a set $\pointset$ of points,
and a distance function $\distance : \pointset \times \pointset \to \reals_{\geq 0}$,
satisfying for all $\pointvar_0,\pointvar_1,\pointvar_2 \in \pointset$:
(i) the coincidence axiom,
$\distance(\pointvar_0,\pointvar_1) = 0 \iff \pointvar_0 = \pointvar_1$;
(ii) symmetry, 
$\distance(\pointvar_0,\pointvar_1) = \distance(\pointvar_1,\pointvar_0)$; and
(iii) the triangle inequality
$\distance(\pointvar_0,\pointvar_1) \leq \distance(\pointvar_0,\pointvar_2) + \distance(\pointvar_2,\pointvar_1)$.
\end{definition}

\subsection{Network Optimization (on Graphs)}


\providecommand{\flownet}{{\mathcal N}}
\providecommand{\newflownet}{\tilde\flownet}
\providecommand{\flowmap}{f}
\providecommand{\optflow}{\flowmap^*}
\providecommand{\newflowmap}{\tilde\flowmap}

\newcommand{\supplyvar}{b}

\newcommand{\flowcost}{J}
\newcommand{\newflowcost}{\tilde\flowcost}

\renewcommand{\pathvar}{P}
\renewcommand{\pathset}{{\mathcal P}}
\renewcommand{\cyclevar}{Q}
\newcommand{\cycleset}{{\mathcal Q}}

%
%

\begin{definition}[Vertex Supplies]
Given a di-graph $(\nodeset,\arcset)$,
a \emph{supply mapping} is a function $\supplyvar : \nodeset \to \reals$.
A supply mapping associates with each vertex $u \in \nodeset$ a \emph{supply} $\supplyvar(u) \in \reals$.
If $\supplyvar(u) > 0$, then $u$ is called a \emph{supply node};
if $\supplyvar(u) < 0$, then $u$ is called a \emph{demand node}, with ``demand'' $-\supplyvar(u) > 0$;
if $\supplyvar(u) = 0$, then $u$ is called a \emph{transshipment node}.
(We assume that $\sum_{ u \in \nodeset } \supplyvar(u) = 0$.)
\end{definition}
\begin{definition}[Flow Network]
A \emph{flow network} $\flownet$ is a tuple $( (\nodeset,\arcset), \supplyvar )$
of a digraph, or network, $(\nodeset,\arcset)$ and a supply mapping $\supplyvar$.
\end{definition}
\begin{definition}[Admissible Flow]
Given a flow network $\flownet = ( (\nodeset,\arcset), \supplyvar )$,
a \emph{flow} is any non-negative mapping $f : \edgeset \to \reals_{\geq 0}$.
An \emph{admissible flow} is a flow satisfying
\begin{align}
  & \supplyvar(u) + \sum_{ \arcvar \in \arcset \ : \ \arcvar^+ = u } f(\arcvar)
  = \sum_{ \arcvar \in \arcset \ : \ \arcvar^- = u } f(\arcvar)
    & ( u \in \nodeset ).
  \label{eq:conservation}
\end{align}
We call~\eqref{eq:conservation} the \emph{flow conservation constraints}.
We use standard shorthand notation $\flowmap \in \flownet$
(e.g., see~\cite{ahuja1993network})
to say $\flowmap$ is admissible by flow network $\flownet$.
%
\end{definition}

%
%

\newcommand{\edgecost}{c}
\newcommand{\edgecosts}{{\bf c}}

\begin{definition}[Flow Costs]
Let $\flownet$ be a flow network and let
$\edgecosts$ be a collection associating to each edge $\arcvar \in \arcset$
a \emph{cost function} $\edgecosts( \argholder ; \arcvar )$.
We define the total cost of a flow $\flowmap \in \flownet$
[under edge costs $\edgecosts$] as
\begin{equation} \label{eq:total flow cost}
  \flowcost(\flowmap ; \edgecosts )
	\doteq
  \sum_{ \arcvar \in \edgeset } \edgecosts( \flowmap(\arcvar) ; \arcvar ).
\end{equation}
%
\end{definition}

%

\begin{definition}[Minimum-Cost Admissible Flow]
Given a flow network $\flownet$ and edge costs $\edgecosts$,
an admissible flow $\flowmap \in \flownet$ is a
\emph{minimum-cost admissible flow} if
$\flowcost(\flowmap ; \edgecosts) \leq \flowcost(\newflowmap; \edgecosts)$
for all admissible flows $\newflowmap \in \flownet$.
\end{definition}

\newcommand{\edgewts}{{\bf w}}
\newcommand{\edgewt}{w}
\newcommand{\newedgewts}{{\tilde\edgewts}}

\begin{definition}[Linearly ``Weighted'' Flow Costs]
If edge costs $\edgecosts$ have the property
that $\flowcost(\flowmap ;\edgecosts)$ is \emph{linear} in $\flowmap$,
i.e., for some \emph{edge weights}
$\edgewts : \arcset \to \reals_{\geq 0}$,
$\flowcost( \flowmap ;\edgecosts) = \sum_{\arcvar \in \arcset} \edgewts(\arcvar) \flowmap(\arcvar)$,
then we write $\flowcost( \cdot ; \edgewts ) \equiv \flowcost( \cdot ; \edgecosts )$.
\end{definition}

\subsection{The Earth Mover's Distance---Properties}

When the domain $\env$ is a finite set,
then the EMD is given by the cost of the optimal solution to:
minimize
over all possible mappings $\gamma : \env^2 \to \reals_{\geq 0}$,
such that
$\sum_{j \in \env} \gamma(i,j) = \measone(i)$ for all $i \in \env$ and
$\sum_{i \in \env} \gamma(i,j) = \meastwo(j)$ for all $j \in \env$,
the cost
$\sum_{i,j \in \env} \gamma(i,j) \distance(i,j)$.
%
\begin{remark}[Network flow interpretation of EMD]
\label{remark:network flow interp}
Equivalently, the EMD is the cost of the minimum-cost admissible flow on the \emph{distance network} over $\env$---%
the complete, directed graph on $\env$ where each edge $(i,j)$ has weight $\distance(i,j)$---%
with supplies $\supplyvar(\cdot) := \measone(\cdot) - \meastwo(\cdot)$.
(This interpretation is valid so long as $\distance$ is a true distance metric.)
\end{remark}
The generalization of such notions to continuous metric spaces (e.g., Euclidean $\reals^d$)
requires measure-theoretic considerations resulting in~\eqref{eq:EMDdef}.
%

The EMD has a quite general shift-invariance property which will be exploited crucially in this paper:
\begin{prop}[Additive invariance of EMD] \label{lemma:EMD add invariant}
Let $\measone$, $\meastwo$, and $\tilde\measvar$ be three distributions
over a finite domain $\env$.
Then $\Wass(\measone+\tilde\measvar,\meastwo+\tilde\measvar) = \Wass(\measone,\meastwo)$.
\end{prop}
\begin{proof}
The proof is simply by Remark~\ref{remark:network flow interp}
and observing that the supply mapping $\supplyvar(\cdot) = \measone(\cdot) - \meastwo(\cdot)$
is invariant to the addition.
\end{proof}
Proposition~\ref{lemma:EMD add invariant} formalizes the intuitive notion that adding
the same ``offset'' to two histograms should not affect the cost of transforming one into the other.
%
Now let the symbol $\preceq$ denote a \emph{vector inequality},
such that in finite domains $\env$, $\measvar' \preceq \measvar$ means that
$\measvar'(i) \leq \measvar(i)$ for all $i \in \env$.
(Such inequality generalized readily.)
\begin{corollary}[Subtractive invariance of EMD] \label{lemma:EMD subtract invariant}
Let $\measone$, $\meastwo$, and $\tilde\measvar$ be three distributions
over a finite domain $\env$,
with $\tilde\measvar \preceq \measone$ and $\tilde\measvar \preceq \meastwo$.
Then $\Wass(\measone - \tilde\measvar,\meastwo - \tilde\measvar) = \Wass(\measone,\meastwo)$.
\end{corollary}
\begin{proof}
The proof is simply by observing that since
$\tilde\measvar \preceq \measone$ and $\tilde\measvar \preceq \meastwo$, then
$\Wass(\measone,\meastwo)
	= \Wass( (\measone-\tilde\measvar)+\tilde\measvar, (\meastwo-\tilde\measvar)+\tilde\measvar )$.
Applying Prop~\ref{lemma:EMD add invariant} obtains the corollary.
\end{proof}
Prop.~\ref{lemma:EMD add invariant} and Corollary~\ref{lemma:EMD subtract invariant}
generalize fully, but the proofs are beyond the scope of this paper.
The finite-version proofs have been presented for the sake of intuition.

\section{The Geometry of Road Networks}
\label{sec:roadnet geometry}

\newcommand{\coorddist}{d}
\newcommand{\roadobj}{{\bf x}}

A roadmap can be described in terms of a set of lines or curves
connected together into a particular pattern by their endpoints;
the distance between points on a roadmap is the minimum distance by which
a particle (or vehicle) could reach one point from the other
while constrained to travel on the curves, or \emph{roads}.
It is common practice, e.g., by modern postal services,
to represent the topology of a roadmap using an undirected weighted graph or multi-graph $(\roadverts,\roadset)$,
possibly with loops,
where the edges $\roadset$ correspond to roads in the roadmap and are labeled with \emph{lengths}, and
the vertices $\roadverts$ describe their interconnections.
Another common practice is to attach to such graph a coordinate system:
Given a fixed orientation of the roadmap graph,
every point on the roadmap continuum can be described unambiguously by a tuple, or \emph{address} $(\roadvar,y)$,
of a road $\roadvar \in \roadset$ and a real-valued coordinate $\coordvar$ between zero and the length $\roadlen_\roadvar$ of $\roadvar$.
%
There is an intuitive notion of ``roadmap distance'' between points described by such addresses,
arising from two basic assertions:
(i) there is a path between any two points on the same road, of length equal to
the difference between their address coordinates;
(ii) there is a special point for every roadmap vertex $u \in \roadverts$ which is on all the roads adjacent to $u$ simultaneously .
%

In this paper, we assume an orientation of the road system has been fixed, so that $\roadset$ is directed.
If an address $\coordvar$ refers to a road $\roadvar$, then
we say $\coordvar \in \roadvar$.
If the coordinate of $\coordvar$ is $y=0$ or $y=\roadlen_\roadvar$, then the coordinate also corresponds to a road endpoint
(the tail or the head, respectively):
if $y = 0$, then we say $\coordvar \in \roadvar^-$;
if $y = \roadlen_\roadvar$, then we say $\coordvar \in \roadvar^+$.

\newcommand{\coordToPoint}{\hat\roadnet}

\begin{definition}[Road Network]
A \emph{road network} is a metric space $(\roadnet,\distance)$,
with point set $\roadnet \doteq \roadverts \cup \{ (\roadvar,y) \ : \ \roadvar\in\roadset, 0 < y < \roadlen_\roadvar \}$
for some representation $(\roadverts,\roadset,\roadlen)$, 
such that
for every pair of points $(\pointvar_1,\pointvar_2) \in \roadnet^2$,
the distance $\distance(\pointvar_1,\pointvar_2)$ is equal to the shortest roadmap distance
between addresses of $\pointvar_1$ and $\pointvar_2$, respectively.
%
\end{definition}



%

%

\begin{prop} \label{prop:roadnet_complete_separable}
Road networks are complete and separable.
\end{prop}
\begin{proof}
The point set of a road network is composed of
(i) a set of open intervals, all disjoint (the roads), and
(ii) another \emph{finite} point set, i.e. $\roadverts$.
The only limit points missing from the collection of roads are the interval boundaries,
which are finite in number and are ``filled in'' by (ii).
Injection of the finite set of points $\roadverts$ cannot introduce \emph{new} limit points, therefore $\roadnet$ is complete.
$\roadnet$ is separable because it is a finite union of separable components.
\end{proof}
%
\begin{corollary}
The Earth Mover's distance~\eqref{eq:EMDdef} is well defined on any road network $(\roadnet,\distance)$, $\roadnet =: \env$.
\end{corollary}
\begin{proof}
A road network $(\roadnet,\distance)$, being a complete and separable metric space,
is therefore a Polish metric space, and also a Radon space.
The Earth Mover's distance is the same as the $1$-Wasserstein distance,
which is defined for all Radon spaces~\cite[Ch.7]{ambrosio2005gradient}.
\end{proof}
%

\subsection{Probability and Road Networks}

\newcommand{\outcomes}{\Omega}
\newcommand{\borelfield}{{\mathcal B}}
\newcommand{\Lebmeas}{\lambda}

\newcommand{\OutcomeMap}{\Delta}

\newcommand{\denset}{{\boldsymbol\den}}		
\newcommand{\densetone}{\denset^\srcTag}
\newcommand{\densettwo}{\denset^\trgTag}

%
%
%

Given a road network metric space $(\roadnet,\distance)$,
let $\borelfield$ denote the \emph{Borel sets} ($\sigma$-algebra)
generated by all the open sets in the topology defined on $\roadnet$ by $\distance$.
Let $\sigfield$ denote the corresponding \emph{Lebesgue measurable sets}.

\begin{definition}[Absolute continuity of measure]
A measure $\measvar$ over a measurable roadmap $(\roadnet,\sigfield)$ is
\emph{absolutely continuous} if there exists a Lebesgue measurable mapping
$\den_\measvar$
such that $\measvar(A) = \int_A \den_\measvar(\pointvar) \ d\pointvar$
for all $A \in \sigfield$;
equivalently, if there exists a set of mappings 
$\denset =: \{ \den_\roadvar : \reals \to \reals \}_{ \roadvar \in \roadset }$
such that
\begin{equation}
  \measvar(A) =
    \sum_{ \roadvar \in \roadset }
    \int_{ y \ : \ (\roadvar,y) \in A }
      \den_\roadvar(y) \ dy.
\end{equation}
We call the components of $\denset$ the \emph{road densities}.
\end{definition}

\begin{assump} \label{assump:nicemeasures}
We restrict our attention to finite, absolutely continuous probability distributions (unity total measure) on road networks,
with Lipschitz road densities.

\end{assump}

In this paper,
we will denote by
$\densetone = \{ \denone_\roadvar \}_{\roadvar \in \roadset}$ and
$\densettwo = \{ \dentwo_\roadvar \}_{\roadvar \in \roadset}$
the densities of distributions $\measone$ and $\meastwo$, respectively.

\begin{definition}[Cumulative density function]
\label{def:cumden}
Given a Lipschitz density function $\den : [0,\roadlen] \to \reals_{\geq 0}$,
let
\[
  \cumden( y ; \den ) \doteq \int_0^{ y } \den(y') \ dy'.
\]
$\cumden( \cdot ; \den )$ is called the cumulative density function (cdf) of $\den$, and
for $\den$ Lipschitz, $\cumden$ is continuous and non-decreasing.
Let $\invcumden( x ; \den ) \doteq \inf \{ y : \ \cumden( y ; \den ) \geq x \}$.
$\invcumden( \cdot ; \den )$ is called the \emph{inverse} cumulative density function, because
$\cumden( \invcumden( x ; \den ) ; \den ) = x$ for all $x \in [0,\roadlen]$.
\end{definition}

%

\section{The Earth Movers Distance on Road Networks}
\label{sec:main result}

\newcommand{\wasstag}{{\rm EXACT}}
\newcommand{\wassnet}{\flownet^\wasstag}
\newcommand{\wassnodes}{\nodeset^\wasstag}
\newcommand{\wassedges}{\arcset^\wasstag}
\newcommand{\wasssupply}{\supplyvar^\wasstag}
\newcommand{\wasscosts}{\edgecosts^\wasstag}

\newcommand{\decisionedges}{\edgeset^{\rm Dec}}
\newcommand{\routingedges}{\edgeset^{\rm Rte}}
\newcommand{\routingwts}{\edgewts^{\rm Rte}}

\newcommand{\supplyset}{S}
\newcommand{\demandset}{D}
\newcommand{\xshipset}{T}

\newcommand{\costform}{q}

\newcommand{\reverseden}{\chi}
\newcommand{\reversedenset}{ {\boldsymbol\reverseden} }

\newcommand{\shortestpath}{{\rm Dijkstra}}

\subsection{Formulation}
\label{sec:main result construction}

A network optimization problem instance is the pair $(\flownet,\edgecosts)$
of a flow network $\flownet$ and edge costs $\edgecosts$.
In this section, we provide a method
to construct a finite-dimensional, convex problem instance
whose optimal solution has cost equal to $\Wass(\measone,\meastwo)$, where
$\measone$ and $\meastwo$ are input distributions
over a roadmap $\roadnet$ described by $(\roadverts,\roadset)$.
We will refer to our particular construction of $\flownet$ as the \emph{Wasserstein network}.

\subsubsection{Technical Assumptions}

\newcommand{\denmin}{\underline\den}
\newcommand{\newdenone}{{\tilde\den}^\srcTag}
\newcommand{\newdentwo}{{\tilde\den}^\trgTag}

\begin{assump} \label{assump:pointwise_posmutex}
For technical reasons, we assume that the supports of $\measone$ and $\meastwo$ are disjoint;
e.g., it holds that $\denone_\roadvar(y) \times \dentwo_\roadvar(y) = 0$
for all $\roadvar\in\roadset$, $y \in [ 0, \roadlen_\roadvar ]$.
\end{assump}
Assumption~\ref{assump:pointwise_posmutex} is actually without loss of generality,
since one may subtract the $\min$ of $\measone$ and $\meastwo$ without altering the EMD
(Corollary~\ref{lemma:EMD subtract invariant}, generalized).
%
%
\begin{assump} \label{assump:road_posmutex}
Let $\measvar(\roadvar)$ denote the total probability of road $\roadvar$ under distribution $\measvar$.
We assume that $\measone(\roadvar) \times \meastwo(\roadvar) = 0$ for all $\roadvar\in\roadset$;
that is, only one of the input distributions may be positive on any given road.
\end{assump}
Assumption~\ref{assump:road_posmutex} supercedes
Assumption~\ref{assump:pointwise_posmutex}, but it is also quite benign.
Roads satisfying Assumptions~\ref{assump:nicemeasures} and~\ref{assump:pointwise_posmutex}
but not~\ref{assump:road_posmutex}
can be ``cracked''---by injecting additional vertices---%
such that Assumption~\ref{assump:road_posmutex} becomes satisfied.
Such insertions do not alter the essential structure of the road network,
e.g., shortest-path distances are preserved.

\subsubsection{Instance Construction}

In order to distinguish our main (exact) construction from others in the paper,
we will denote the flow network
$\wassnet =: ( (\wassnodes,\wassedges), \wasssupply )$
and the edge costs $\wasscosts$.
The construction of the network $\wassnet$ is as follows:
We begin with both $\wassnodes$ and $\wassedges$ empty.
Then, we insert into $\wassnodes$ the whole collection of roads and interchanges
$\roadset \cup \roadverts$.
While the roads in $\roadset$ are \emph{edges} of the roadmap,
they are treated simply as vertices in $\wassnet$.
Let $\supplyvar(\roadvar) := \measone(\roadvar) - \meastwo(\roadvar)$ be called the \emph{surplus} of road $\roadvar$.
The supplies associated with $\wassnodes$ will be
\begin{equation}
	\wasssupply(u) := 
		\text{
			$\supplyvar(u)$
			for $u \in \roadset$; 
			$0$ for $u \in \roadverts$.
		}
\end{equation}

Let us create a partition of the set of roads $\roadset$.
For any road $\roadvar$, if $\measone(\roadvar) > 0$, then we call it a \emph{supply road};
if $\meastwo(\roadvar) > 0$, then we call it a \emph{demand road}.
According to Assumption~\ref{assump:road_posmutex},
a road may be \emph{either} a supply road or a demand road, but not both;
if it is neither, i.e., $\measone(\roadvar) = \meastwo(\roadvar) = 0$,
then we call it a \emph{transshipment road}.
We can write the set of supply roads as $\supplyset := \{ \roadvar \in \roadset : \ \supplyvar(\roadvar) > 0 \}$,
demand roads as $\demandset := \{ \roadvar \in \roadset : \ \supplyvar(\roadvar) < 0 \}$,
and transshipment roads as $\xshipset := \{ \roadvar \in \roadset : \ \supplyvar(\roadvar) = 0 \}$.

For each supply road $\roadvar \in \supplyset$,
we insert directed edges $(\roadvar,\roadvar^-)$ and $(\roadvar,\roadvar^+)$
into $\wassedges$.
Even in the case $\roadvar^- = \roadvar^+$,
these notations will denote two separate and distinct edges
(though, in such case, with the same endpoints);
therefore, note that $\wassnet$ could be a multi-graph.
We will use the alias $\roadconnleft$
to refer to $(\roadvar,\roadvar^-)$ and
$\roadconnright$ to refer to $(\roadvar,\roadvar^+)$.
For each demand road $\roadvar \in \demandset$,
we add the edges $(\roadvar^-,\roadvar)$
and $(\roadvar^+,\roadvar)$
into $\wassnodes$;
such edges are also always distinct,
and are also given aliases $\roadconnleft$ and $\roadconnright$ (respectively),
though they have the opposite direction.
$\wassedges$ now contains the \emph{decision edges};
let us denote this set $\decisionedges$.

The costs on the decision edges are as follows.
Let
\[
	\den_\roadvar \doteq
	\begin{cases}
		\denone_\roadvar,		& \text{if $\roadvar \in \supplyset$}		\\
		\dentwo_\roadvar,		& \text{if $\roadvar \in \demandset$},
	\end{cases}	
	\qquad \text{and } \qquad
	\reverseden_\roadvar( x ) \doteq \den_\roadvar( \roadlen_\roadvar - x )
	\qquad \text{for all $\roadvar \in \supplyset \cup \demandset$}.
\]
Let
\begin{equation} \label{eq:costform}
  \costform(x; \den ) \doteq \int_{y=0}^{ \invcumden( x; \den ) } \den(y) \ y \ dy
  \qquad \text{for any $\den$}.
\end{equation}
Then
\begin{align}
\wasscosts( \, \cdot \, ; \roadconnleft ) &:= \costform( \, \cdot \, ; \den_\roadvar ),
	& & \text{and}	
	\label{eq:tail conn cost}		\\			
\wasscosts( \, \cdot \, ; \roadconnright ) &:= \costform( \, \cdot \, ; \reverseden_\roadvar ),
	& &
\text{for all $\roadvar \in \supplyset \cup \demandset$.}
	\label{eq:head conn cost}
\end{align}

%
%
%

Now let $\routingedges$ denote a set of \emph{routing edges}:
$\routingedges$ contains
one edge in each direction
between any pair $u,v \in\roadverts$,
if $u \neq v$ and they are connected by some $\roadvar\in\roadset$;
such edge has linear cost
with weight
$\routingwts( (u,v) )$
equal to the length of the shortest such road.
We insert all of the routing edges into $\wassedges$.

Figure~\ref{fig:wassnet_construction_example} shows a simple road network
with roads ``north'' (N), ``east'' (E), ``south'' (S), and ``west'' (W),
and Figure~\ref{fig:wassnet_construction_network} shows the corresponding Wasserstein network.
\begin{figure}[h!]
\centering
\hfill
\subfigure[A square road network with roads: North (N), East (E), South (S), and West (W).]{

\begin{tikzpicture}
  \tikzstyle{vertex}=[circle,thick,draw=blue!75,fill=blue!20,minimum size=6mm]
  \tikzstyle{roadedge}=[thick]
  \tikzstyle{roadnode}=[rectangle,thick,draw=black!75,
  			  fill=black!20,minimum size=4mm]

  
  \newcommand{\myscale}{1}
  \node [vertex] (1) at (-\myscale,\myscale) {$1$} ;
  \node [vertex] (2) at (\myscale,\myscale) {$2$} ;
  \node [vertex] (3) at (\myscale,-\myscale) {$3$} ;
  \node [vertex] (4) at (-\myscale,-\myscale) {$4$} ;
  
  \path
  	(1) edge node [above] {N $\in \demandset$} (2)
  	(2) edge node [right] {E $\in \supplyset$} (3)
  	(3) edge node [below] {S $\in \supplyset$} (4)
  	(4) edge node [left]  {W $\in \demandset$}	(1);
\end{tikzpicture} 
	\label{fig:wassnet_construction_example}
}
\hfill
\subfigure[The Wasserstein network resulting from the roadmap in Figure~\ref{fig:wassnet_construction_example}.]{
	\begin{tikzpicture}[scale=.8]
  \tikzstyle{vertex}=[circle,thick,draw=blue!75,fill=blue!20,minimum size=6mm]
  \tikzstyle{roadnode}=[rectangle,thick,draw=black!75,
  			  fill=black!20,minimum size=4mm]
  			  
  \tikzstyle{decisionedge}=[very thick,->]
  \tikzstyle{routingedge}=[->,bend right]
  \tikzstyle{inuse}=[thick]
  
  \newcommand{\myscale}{1}
  \node [vertex] (1) at (-\myscale,\myscale) {1} ;
  \node [vertex] (2) at (\myscale,\myscale) {2} ;
  \node [vertex] (3) at (\myscale,-\myscale) {3} ;
  \node [vertex] (4) at (-\myscale,-\myscale) {4} ;
  \newcommand{\rfactor}{2.5}
  \node [roadnode] (N) at (0,\rfactor*\myscale) {N} ;
  \node [roadnode] (E) at (\rfactor*\myscale,0) {E} ;
  \node [roadnode] (S) at (0,-\rfactor*\myscale) {S} ;
  \node [roadnode] (W) at (-\rfactor*\myscale,0) {W} ;
  
  \newcommand{\labelstyle}{\tiny}
  \path
  	(1) edge [decisionedge] (N)
  	(2) edge [decisionedge] node [above right] {} (N)
  	(E) edge [decisionedge] node [above right] {} (2)
  	(E) edge [decisionedge] node [below right] {} (3) 
  	(S) edge [decisionedge] (3) 
  	(S) edge [decisionedge] node [below left] {} (4) 
  	(1) edge [decisionedge] node [above left] {} (W) 
  	(4) edge [decisionedge] node [below left] {} (W) ;
  	%
  	
  \path
  	(1) edge [routingedge] (2)
  	(2) edge [routingedge] node [above] {} (1)
  	(2) edge [routingedge] (3)
  	(3) edge [routingedge] (2)
  	(3) edge [routingedge] node [above] {} (4)
  	(4) edge [routingedge] (3)
  	(4) edge [routingedge] (1)
  	(1) edge [routingedge] (4) ;
\end{tikzpicture}
	\label{fig:wassnet_construction_network}
}
\hfill\null
\caption{A simple road network and the resulting ``Wasserstein'' flow network.}
\label{fig:wassnet_construction}
\end{figure}
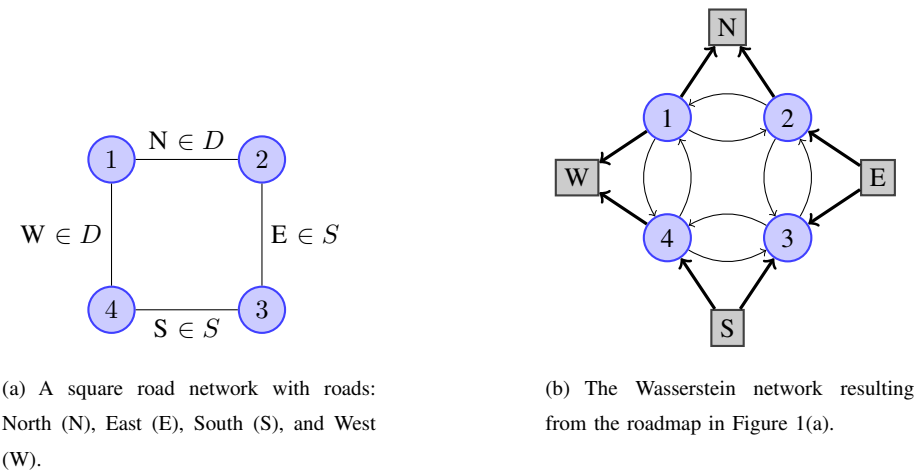
E and S are supply roads and N and W are demand roads;
therefore, notice that the decision edges---%
shown by thick lines---%
point \emph{out of} E and S,
while they point \emph{into} N and W.
Each road also contributes a pair of routing edges to the network.

\subsubsection{Main Result}

The main result of the paper is Theorem~\ref{thm:roadnet emd equation} below:
\begin{theorem}
\label{thm:roadnet emd equation}
Let $(\roadnet,\distance)$ be a road network environment
described by $(\roadverts,\roadset)$,
and let $\measone$ and $\meastwo$ be two finite measures over $\roadnet$,
satisfying Assumptions~\ref{assump:nicemeasures} and~\ref{assump:road_posmutex},
with \emph{equal} total measure.
Then
\begin{equation} \label{eq:roadnet emd equation}
	\min_{ \flowmap \in \wassnet } \ \flowcost( \flowmap ; \wasscosts )
		=
	 \Wass(\measone,\meastwo).
\end{equation}
\end{theorem}
Note that the left-hand side of~\eqref{eq:roadnet emd equation} is a finite-dimensional optimization problem with only \emph{linear} equality and inequality constraints.
The proof of Theorem~\ref{thm:roadnet emd equation} requires many intermediate results that we have not yet established;
we defer the proof until the end of the paper.


\subsection{Convexity of the EMP Objective}

The next result is crucial to show that the formulation of Theorem~\ref{thm:roadnet emd equation} is explicit.
\begin{theorem} \label{thm:convexity_of_objective}
The objective function $\flowcost( \argholder ;\wasscosts)$ is convex over $\flowmap \in \wassnet$.
\end{theorem}

Theorem~\ref{thm:convexity_of_objective} follows as an easy consequence of the next proposition, since
sums of convex functions are convex.
\begin{prop} \label{prop:costform convex}
For every Lipschitz function $\den : [0,\roadlen] \to \reals_{\geq 0}$,
$\costform( \argholder ;\den)$ is Lipschitz continuous and convex over the interval $[0,\cumden(\roadlen;\den)]$.
\end{prop}
\begin{proof}
The absolute difference $\left| \costform(x';\den) - \costform(x;\den) \right|$ can be written as
\[
\left| \costform(x';\den) - \costform(x;\den) \right|
	= \left| \int_{ \invcumden(x;\den) }^{ \invcumden(x';\den) } \den(y) \ y \ dy \right|.
\]
Because
the range of $\invcumden(\cdot;\den)$ is $[0,\roadlen]$,
we have $y \leq \roadlen$ over the whole integral range.
Therefore,
\[
	\left| \int_{ \invcumden(x;\den) }^{ \invcumden(x';\den) } \den(y) \ y \ dy \right|
	\leq
	\roadlen \left| \int_{ \invcumden(x;\den) }^{ \invcumden(x';\den) } \den(y) \ dy \right|
	= \roadlen \left| x' - x \right|.
\]
Then $\costform( \argholder ;\den)$ is Lipschitz, because
for every $x, x' \in [0,\cumden(\roadlen;\den)]$
\[
	\frac{
		\left| \costform(x';\den) - \costform(x;\den) \right|
	}{
		| x' - x |
	}
	\leq \roadlen.
\]

To show that $\costform(\argholder;\den)$ is convex, we observe that
for all $x_0,x \in [0,\cumden(\roadlen;\den)]$ it holds
\begin{equation} \label{eq:costform tangents}
\costform(x;\den) \geq \costform(x_0;\den) + \invcumden(x_0;\den) [ x - x_0 ],
\end{equation}
i.e.,
there is a tangent line at every point $x_0 \in [0,\cumden(\roadlen;\den)]$
with $\costform(\argholder;\den)$ lying entirely above it.
Such functions are known to be convex, e.g.,
by convexity of epigraphs which can be expressed as the intersection of many linear epigraphs.
To verify~\eqref{eq:costform tangents}, one can write
\[
\begin{aligned}
	\costform(x;\den) - \costform(x_0;\den)
		&= \int_{ \invcumden(x_0;\den) }^{ \invcumden(x;\den) } \den(y) \ y \ dy	\\
		&\geq \invcumden(x_0;\den) \int_{ \invcumden(x_0;\den) }^{ \invcumden(x;\den) } \den(y) \ dy
		= \invcumden(x_0;\den) [ x - x_0 ].
\end{aligned}
\]
%
%

\end{proof}

While $\costform$ may be difficult to obtain in analytical form, except in special cases,
\eqref{eq:costform tangents} demonstrates that $\invcumden$ is everywhere in its subgradient.
Gradient and subgradient methods are at the heart of modern algorithms
for constrained optimization of general convex functions,
and Theorem~\ref{thm:convexity_of_objective} provides a certificate that $\costform( \argholder ;\den)$ is convex
regardless of the density function $\den$.
Therefore, provided one has access to an evaluable expression (or ``\emph{circuit}'') for $\invcumden$,
then our formulation is highly amenable to modern convex optimization techniques.

%
%

\subsubsection{Road-wise Uniform Density}
\providecommand{\realdensity}{\rho}
In the special case that all of the road densities are uniform, then
we obtain
$\invcumden( x ; \den ) = x / \realdensity$ and
$\costform( x ; \den_\roadvar ) = \onehalf x^2 / \realdensity$
for each $\roadvar \in \roadset$,
where $\realdensity$ is the constant level of $\den$, or, abusing terminology, its ``density''.
Thus, if the density functions are uniform over all segments,
then the decision edge costs are all \emph{convex quadratic} in $\flowmap$.
The resulting class of network optimization problems can be solved by way of
\emph{quadratic programming} (QP),
a well-studied approach to optimization problems
with convex quadratic objective and linear constraints~\cite[p.152]{Boyd:2004}.

\subsection{Discussion}

Our fully rigorous proof of Theorem~\ref{thm:roadnet emd equation} is quite technical, and
requires several sections of supporting analysis.
However, in the present section
we provide an informal interpretation of the result,
based on the previous ``pile of dirt'' analogy.
%

Consider a single road $\roadvar$ of length $\roadlen$ (see Figure~\ref{fig:decision edge intuition}).
%
\begin{figure}[h!]
\centering
\begin{tikzpicture}
\tikzstyle{roadstyle}=[very thick]
\tikzstyle{roadvertex}=[circle,inner sep=2,fill]
\tikzstyle{dqstyle}=[very thick,black!50]
\newcommand{\roadright}{6}
\newcommand{\roadlevel}{-1}
\draw[roadstyle]
	(0,\roadlevel) node [roadvertex] (roadleft) {}
		node [anchor=east] {$\roadvar^-$}
		-- (\roadright,\roadlevel) node [roadvertex] (roadright) {}
		node [anchor=west] {$\roadvar^+$} ;	
\draw [thick,decoration={brace,amplitude=5,raise=5,mirror},decorate]
	(0,\roadlevel) -- node [below=10,anchor=north west] {$\roadvar$} (\roadright,\roadlevel) ;

\draw[->] (0,0) -- (\roadright,0) node [right] {$y$} ;
\foreach \x in {0,.5,...,5.5} \draw (\x,-.1) -- (\x,.1) ;
\draw[->] (0,0) -- (0,3) ;
\foreach \y in {0,.5,...,2.5} \draw (-.1,\y) -- (.1,\y) ;

\newcommand{\thefunc}[1]{ {( 1.5 + sin(50*#1) )} }
\newcommand{\xsplit}{2.5}

\draw (0,0) plot [domain=0:6] (\x,\thefunc{\x})
	node (density end) [anchor=west] {$\den(y)$} ;	

\draw[dashed] (\xsplit,-.25) -- (\xsplit,3)
	node [above right] {$y^* = \invcumden(x;\den)$} ;

\fill[blue,opacity=.1] (0,0) -- plot [domain=0:\xsplit] (\x,\thefunc{\x}) -- (\xsplit,0) -- cycle ;
\draw node at (1.75,1) {$x$} ;

\fill[green,opacity=.1] (\xsplit,0) -- plot [domain=\xsplit:6] (\x,\thefunc{\x}) -- (6,0) -- cycle ;

\newcommand{\xshiftL}{1}
\newcommand{\xshiftR}{4}
\newcommand{\spanheight}{1.75}
\draw[dqstyle] (\xshiftL,0)
	node [black,below=2] {$y_1 \leq y^*$}
	-- (\xshiftL,\thefunc{\xshiftL}) coordinate (top) ; 
\draw[thick,dashed] (top) -- (\xshiftL,3) node [above=2] {$\den(y_1) dy$} ; 
\draw[thick,dashed,<-] (0,\spanheight) -- node [below] {$y_1$} (\xshiftL,\spanheight) ;
\draw[dqstyle] (\xshiftR,0)
	node [black,below=2] {$y_2 \geq y^*$}
	-- (\xshiftR,\thefunc{\xshiftR}) coordinate (top) ; 
\draw[thick,dashed] (top) -- (\xshiftR,2) node [above=2] {$\den(y_2) dy$} ; 
\draw[thick,dashed,->] (\xshiftR,\spanheight) -- node [below] {$\roadlen-y_2$} (6,\spanheight) ;

\end{tikzpicture}
\caption{
A supply road $\roadvar \in \supplyset$.
The area $x$, under the curve to the left of $y^*$, is transported to $\roadvar^-$.
The area $\measone(\roadvar) - x$ to the right of $y^*$ is transported to $\roadvar^+$.
}
\label{fig:decision edge intuition}
\end{figure}
Suppose that $\roadvar$ is a supply road, whose distribution of commodity has density function $\den$.
Since $\roadvar$ is a supply road, all of the demand is elsewhere in the network.
Therefore, all the available commodity must leave $\roadvar$ by one of its endpoints.
Suppose we wish to transport a quantity $x$ of the commodity via $\roadvar^-$, and
the remaining $\measone(\roadvar)-x$ commodity
via $\roadvar^+$.
If the cost of transportation is proportional to distance traveled,
it is easy to argue that moving the left-most $x$ commodity to $\roadvar^-$
and the remainder to $\roadvar^+$
is optimal (see Figure~\ref{fig:decision edge intuition}).
The boundary separating the left-most $x$ commodity from the remainder lies at coordinate $y^* := \invcumden(x;\den)$.
Applying basic calculus, the cost of this strategy is determined to be
\begin{equation}
\label{eq:roadcost derived}
  \int_0^{y^*} \den(y) \ y \ dy
  + \int_{y^*}^{ \roadlen } \den(y) \ [ \roadlen - y ] \ dy.
\end{equation}
The first and second terms of~\eqref{eq:roadcost derived}
provide the cost functions~\eqref{eq:tail conn cost} and~\eqref{eq:head conn cost}, respectively,
if one interprets $x$ as the flow on decision edge $\roadconnleft$ (i.e., $\roadTL$) and
$\measone(\roadvar)-x$
as the flow on decision edge $\roadconnright$ (i.e., $\roadTR$).
Note that the real-valued quantity $x$ is left as one of the dimensions
of the finite-dimensional optimization problem~\eqref{eq:roadnet emd equation}.
A symmetrical argument can be used to obtain the same cost for transporting commodity \emph{into} a demand road.

It may not be possible for all commodity which leaves some road by one of its endpoints
to supply demand on the interior of an adjacent road.
For example, if the total supply on one road exceeds the total demand of its immediate neighborhood, then
some supply must be assigned outside of this neighborhood.
%
However, let us consider a ``strategy'' in three phases:
First, commodity will be ``accumulated'' at vertices as previously described.
The third phase will be exactly opposite in the sense that
commodity will be ``dispersed'' from the vertices to satisfy demand in the interiors of adjacent roads.
During the middle phase, however,
commodity may be ``re-distributed'', but \emph{strictly} on the vertex set $\roadverts$.
The problem of finding the minimum cost re-distribution schedule
given the two ``vertex-only'' distributions of commodity
(i.e.,
the one immediately after accumulation and the one immediately before dispersion),
can be cast as a traditional minimum-cost flow problem
on the routing edges $\routingedges$ with weights $\routingwts$.
The flow conservation constraints~\eqref{eq:conservation} on $\wassnet$
account for the flow conservation requirements of all three phases simultaneously.
%
It turns out that the \emph{optimal} strategy of this three-phase type is at least as good as any other strategy.
%
%

\subsection{Numerical Example}
\label{sec:main result example}

Let us re-visit the example network in Figure~\ref{fig:wassnet_construction_example} and
assign specific distributions.
Suppose each road is of unit length, and has probability given in Figure~\ref{fig:road example measures}.
\begin{figure}
\centering


\begin{tikzpicture}
  \tikzstyle{vertex}=[circle,thick,draw=blue!75,fill=blue!20,minimum size=6mm]
  \tikzstyle{roadedge}=[thick]
  \tikzstyle{roadnode}=[rectangle,thick,draw=black!75,
  			  fill=black!20,minimum size=4mm]

  \tikzstyle{every label}=[red]
  

  \newcommand{\myscale}{1.1}
  \node [vertex] (1) at (-\myscale,\myscale) {$1$} ;
  \node [vertex] (2) at (\myscale,\myscale) {$2$} ;
  \node [vertex] (3) at (\myscale,-\myscale) {$3$} ;
  \node [vertex] (4) at (-\myscale,-\myscale) {$4$} ;
  
  \path
  	(1) edge node [above] {$\meastwo({\rm N})=\frac{1}{5}$ } (2)
  	(2) edge node [right] {$\measone({\rm E})=\frac{2}{5}$ } (3)
  	(3) edge node [below] {$\measone({\rm S})=\frac{3}{5}$ } (4)
  	(4) edge node [left]  {$\meastwo({\rm W})=\frac{4}{5}$ }	(1);
\end{tikzpicture}
%
\caption{
The roadmap of Figure~\ref{fig:wassnet_construction_example}
labeled with measures $\measone$ and $\meastwo$.
}
\label{fig:road example measures}
\end{figure}
The supply or demand of each road is distributed uniformly over its length.
%

Figure~\ref{fig:roadnet example solution} shows two new copies of the flow network $\wassnet$
first shown in Figure~\ref{fig:wassnet_construction_network}.
%
\begin{figure}
\newcommand{\figurescale}{.8}
\centering
\hfill
\subfigure[
Wasserstein network of Figure~\ref{fig:wassnet_construction_network}
labeled with optimal flows under $\measone$ and $\meastwo$.
]{
\begin{tikzpicture}[scale=\figurescale]
  \tikzstyle{vertex}=[circle,thick,draw=blue!75,fill=blue!20,minimum size=6mm]
  \tikzstyle{roadnode}=[rectangle,thick,draw=black!75,
  			  fill=black!20,minimum size=4mm]
  			  
  \tikzstyle{decisionedge}=[->]
  \tikzstyle{routingedge}=[->,bend right]
  \tikzstyle{inuse}=[thick]
  
  \newcommand{\myscale}{1}
  \node [vertex] (1) at (-\myscale,\myscale) {1} ;
  \node [vertex] (2) at (\myscale,\myscale) {2} ;
  \node [vertex] (3) at (\myscale,-\myscale) {3} ;
  \node [vertex] (4) at (-\myscale,-\myscale) {4} ;
  \newcommand{\rfactor}{2.5}
  \node [roadnode] (N) at (0,\rfactor*\myscale) {N} ;
  \node [roadnode] (E) at (\rfactor*\myscale,0) {E} ;
  \node [roadnode] (S) at (0,-\rfactor*\myscale) {S} ;
  \node [roadnode] (W) at (-\rfactor*\myscale,0) {W} ;
  
  \newcommand{\labelstyle}{\tiny}
  \path
  	(1) edge [decisionedge] (N)
  	(2) edge [decisionedge,inuse] node [above right] {\labelstyle $1/5$} (N)
  	(E) edge [decisionedge,inuse] node [above right] {\labelstyle $1/3$} (2)
  	(E) edge [decisionedge,inuse] node [below right] {\labelstyle $1/15$} (3) 
  	(S) edge [decisionedge] (3) 
  	(S) edge [decisionedge,inuse] node [below left] {\labelstyle $3/5$} (4) 
  	(1) edge [decisionedge,inuse] node [above left] {\labelstyle $2/15$} (W) 
  	(4) edge [decisionedge,inuse] node [below left] {\labelstyle $2/3$} (W) ;
  	%

  \path
  	(1) edge [routingedge] (2)
  	(2) edge [routingedge,inuse] node [above] {\labelstyle $2/15$} (1)
  	(2) edge [routingedge] (3)
  	(3) edge [routingedge] (2)
  	(3) edge [routingedge,inuse] node [above] {\labelstyle $1/15$} (4)
  	(4) edge [routingedge] (3)
  	(4) edge [routingedge] (1)
  	(1) edge [routingedge] (4) ;
\end{tikzpicture}
\label{fig:roadnet example solution flow}
}
\hfill
%
\subfigure[
Wasserstein network of Figure~\ref{fig:wassnet_construction_network}
labeled with cost per edge of optimal flows under $\measone$ and $\meastwo$.
]{
\begin{tikzpicture}[scale=\figurescale]
  \tikzstyle{vertex}=[circle,thick,draw=blue!75,fill=blue!20,minimum size=6mm]
  \tikzstyle{roadnode}=[rectangle,thick,draw=black!75,
  			  fill=black!20,minimum size=4mm]
  			  
  \tikzstyle{decisionedge}=[->]
  \tikzstyle{routingedge}=[->,bend right]
  \tikzstyle{inuse}=[thick]
  
  \newcommand{\myscale}{1}
  \node [vertex] (1) at (-\myscale,\myscale) {1} ;
  \node [vertex] (2) at (\myscale,\myscale) {2} ;
  \node [vertex] (3) at (\myscale,-\myscale) {3} ;
  \node [vertex] (4) at (-\myscale,-\myscale) {4} ;
  \newcommand{\rfactor}{2.5}
  \node [roadnode] (N) at (0,\rfactor*\myscale) {N} ;
  \node [roadnode] (E) at (\rfactor*\myscale,0) {E} ;
  \node [roadnode] (S) at (0,-\rfactor*\myscale) {S} ;
  \node [roadnode] (W) at (-\rfactor*\myscale,0) {W} ;
  
  \newcommand{\labelstyle}{\tiny}
  \path
  	(1) edge [decisionedge] (N)
  	(2) edge [decisionedge,inuse] node [above right] {\labelstyle $1/10$} (N)
  	(E) edge [decisionedge,inuse] node [above right] {\labelstyle $5/36$} (2)
  	(E) edge [decisionedge,inuse] node [below right] {\labelstyle $1/180$} (3) 
  	(S) edge [decisionedge] (3) 
  	(S) edge [decisionedge,inuse] node [below left] {\labelstyle $3/10$} (4) 
  	(1) edge [decisionedge,inuse] node [above left] {\labelstyle $1/90$} (W) 
  	(4) edge [decisionedge,inuse] node [below left] {\labelstyle $5/18$} (W) ;
  	%

  \path
  	(1) edge [routingedge] (2)
  	(2) edge [routingedge,inuse] node [above] {\labelstyle $2/15$} (1)
  	(2) edge [routingedge] (3)
  	(3) edge [routingedge] (2)
  	(3) edge [routingedge,inuse] node [above] {\labelstyle $1/15$} (4)
  	(4) edge [routingedge] (3)
  	(4) edge [routingedge] (1)
  	(1) edge [routingedge] (4) ;

	\node at (3*\myscale,2.25*\myscale) {\labelstyle {\bf Total Cost:} $31/30$} ;
	
\end{tikzpicture}
\label{fig:roadnet example solution cost}
}
\hfill\null
\caption{}
\label{fig:roadnet example solution}
\end{figure}
The network in Figure~\ref{fig:roadnet example solution flow} is labeled with the flows of the optimal network flow solution (obtained by solving a quadratic program).
The network in Figure~\ref{fig:roadnet example solution cost} is labeled with the costs incurred on each edge by the optimal network flow.
The optimal solution has cost equal to $31/30$,
which is therefore the Earth Mover's distance between $\measone$ and $\meastwo$.

Examining the optimal flow provides qualitative insight in addition to the value of the EMD.
In particular, we can observe the following facts:
First, the demand of the north road (N) is supplied entirely by the east road (E).
Second, all of the supply of the south road (S) goes to the west road (W).
Finally, the east road (E) supplies the remaining demand of the west road (W),
however, $1/15$ unit of supply from E reaches W via the clockwise path (E-3-4-W),
while the remaining $2/15$ unit of supply reaches W via the counter clockwise path (E-2-1-W).
%

\section{Simulation Study}
\label{sec:simulation}

In this section we present a simulation study motivated by the work in~\cite{treleaven2013},
demonstrating the role of the EMD
in predicting the
throughput of vehicle sharing systems modeled by
stochastic and dynamic Pickup-and-Delivery problems.

\subsection{Background}

We consider the Dynamic Pickup and Delivery problems (DPDP)
with \emph{stochastic} demands, studied
e.g., in~\cite{Swihart:1999,Waisanen:,pavone2010fundamental,treleaven2013}.
(A survey on the general DPDP can be found in~\cite{berbeglia2010dynamic}.)
%
%
A number $\numveh$ of service vehicles travel in a geometric \emph{workspace} $\workspace$
with unit maximum speed;
the distance between points is measured by a distance function $\distance$.
The vehicles have unlimited range but \emph{unit} capacity, i.e., they can transport at most one object at a time.
Demands arrive randomly into the workspace,
generated according to a time-invariant Poisson process with time intensity $\lambda \in \reals_{>0}$.
A newly arrived demand has an associated pickup location $\randpt$ and an associated delivery location $\altrandpt$, where
the demand data $(\randpt,\altrandpt)$ is independently, identically distributed (i.i.d.)
according to a joint probability distribution $\measvar$.
Each demand must be transported from its pickup location to its delivery location---%
i.e., an empty vehicle must visit the pickup location, followed immediately by the delivery location---%
then it is removed from the system.

\cite{treleaven2013} studied the DPDP in
Euclidean workspaces $\workspace \subset \reals^d$, $d \geq 2$, with
distributions $\measvar$ having
absolutely continuous marginal distributions $\measone$ and $\meastwo$ for $\randpt$ and $\altrandpt$, respectively.
It was shown that
under any ``stabilizing'' routing policy---i.e., one where
the number of demands in the system remains uniformly bounded for all time---%
the average vehicle time dedicated to any demand
satisfies a lower bound
\[
	\liminf_{ t \to +\infty } \numserved_t / t
	\geq
		\overline\servicetime.
\]
Here,
$\numserved_t$ denotes the total number of demands serviced by time $t$ and
\begin{equation} \label{eq:defn servicetime}
	\overline\servicetime
		\doteq \expect_\measvar \left\{ \distance( \randpt, \altrandpt ) \right\} + W( \measone, \meastwo ).
\end{equation}
%
%
Using this result, the authors proved Theorem~\ref{thm:stability condition}.
\begin{theorem}[Stability of the DPDP]
\label{thm:stability condition}
Defining the \emph{system utilization} (a fraction) as
$\utilization \doteq \arrivalrate \overline\servicetime / \numveh$,
the condition $\utilization < 1$ is both necessary and sufficient
to ensure the existence of 
a vehicle-routing policy by which
the expected number of demands in the system remains uniformly bounded for all time, i.e.,
it does not grow unbounded.
\end{theorem}
%

Similar results had been proved previously in~\cite{Swihart:1999} and~\cite{pavone2010fundamental}, but
mysteriously \emph{without} the EMD term.
The reason is that in every previous study it was either implicit or assumed that $\measone = \meastwo$,
in which case it happens that $\Wass(\measone,\meastwo) = 0$.
However, in any case when the marginal distributions of $\measvar$ are different---even \emph{slightly}---%
the stability condition reveals the additional Earth Mover's distance term.

%

\subsection{Experiment Design}
\providecommand{\simhorizon}{T}

Our experiment is similar to one in~\cite{treleaven2013}, which 
measures the critical arrival rate $\arrivalrate^*$
separating \emph{stabilizable} arrival rates from \emph{unstabilizable} ones,
given a fixed setting of the other system parameters.
We will not re-derive~\eqref{eq:defn servicetime} or~Theorem~\ref{thm:stability condition} for roadmaps in this paper.
Doing so involves a trivial retracing of the logic in~\cite{treleaven2013}, and
yields little or no new insight.
%
The main insight of the experiment is as follows:
Let $\pi$ be a routing policy for a $\numveh$-vehicle DPDP which
is stabilizing for all $\arrivalrate < \arrivalrate^*$
(and satisfies some technical ``fairness'' conditions).
Then we run the DPDP system with arrival rate $\arrivalrate > \arrivalrate^*$ and
operating under $\pi$.
Since $\arrivalrate > \arrivalrate^*$, the number of outstanding demands in the system grows unbounded.
However, we can expect the policy to service demands at an average rate approaching $\arrivalrate^*$
(i.e., the fastest rate under $\pi$)
as demands build up in the workspace.
Thus, we can estimate $\arrivalrate^*$, e.g.,
by computing $\numserved_\simhorizon / \simhorizon$
after time $\simhorizon$ sufficiently large.

Our simulations are of a ``gated'', multi-vehicle, nearest-neighbor policy (gated $\numveh$-NN).
A gated policy is one that completes in order a sequence of demand ``batches'', where
each batch consists of all the demands that arrived while
the previous batch was being worked on.
Within a particular batch, a vehicle $i$'s $k$th demand is the one---%
among all demands not yet assigned to any vehicle at the time when $i$'s $(k-1)$th demand was delivered---%
whose pickup location was nearest to the location of $i$.
Although a proof that such policy is stabilizing for all $\arrivalrate < \arrivalrate^*$ is currently not available,
it has been observed that nearest neighbor policies have good performance for a variety of vehicle routing problems.

\subsection{Results and Discussion}

The simulation experiment was repeated for fifty (50) randomly-generated scenarios,
each characterized by
(i) a random, connected roadmap $\roadnet$ of $1$--$10$ roads,
(ii) a random demand distribution $\measvar$ (with random but constant density per pair of roads), and
(iii) a randomly sized fleet of between $1$--$5$ unit speed vehicles.
The minimum average service time $\overline\servicetime$ was predicted using~\eqref{eq:defn servicetime};
$\Wass$ was computed by solving a QP, and
the expected pickup-to-delivery distance
was computed using another method which is the subject of a future paper
(Monte Carlo averaging is a viable option).
The critical rate $\arrivalrate^*$ was computed by $\numveh / \overline\servicetime$.
In each case, the arrival rate simulated was $2\arrivalrate^*$
(exceeding theoretical capacity by 100\%), and
the simulation was run for $\simhorizon=1000$ time.
Figure~\ref{fig:service_time_results} shows a very strong corroboration
between the computed and empirical per-demand average service times $\overline\servicetime$.
\begin{figure}
\centering
\includegraphics[
width=.4\linewidth
]{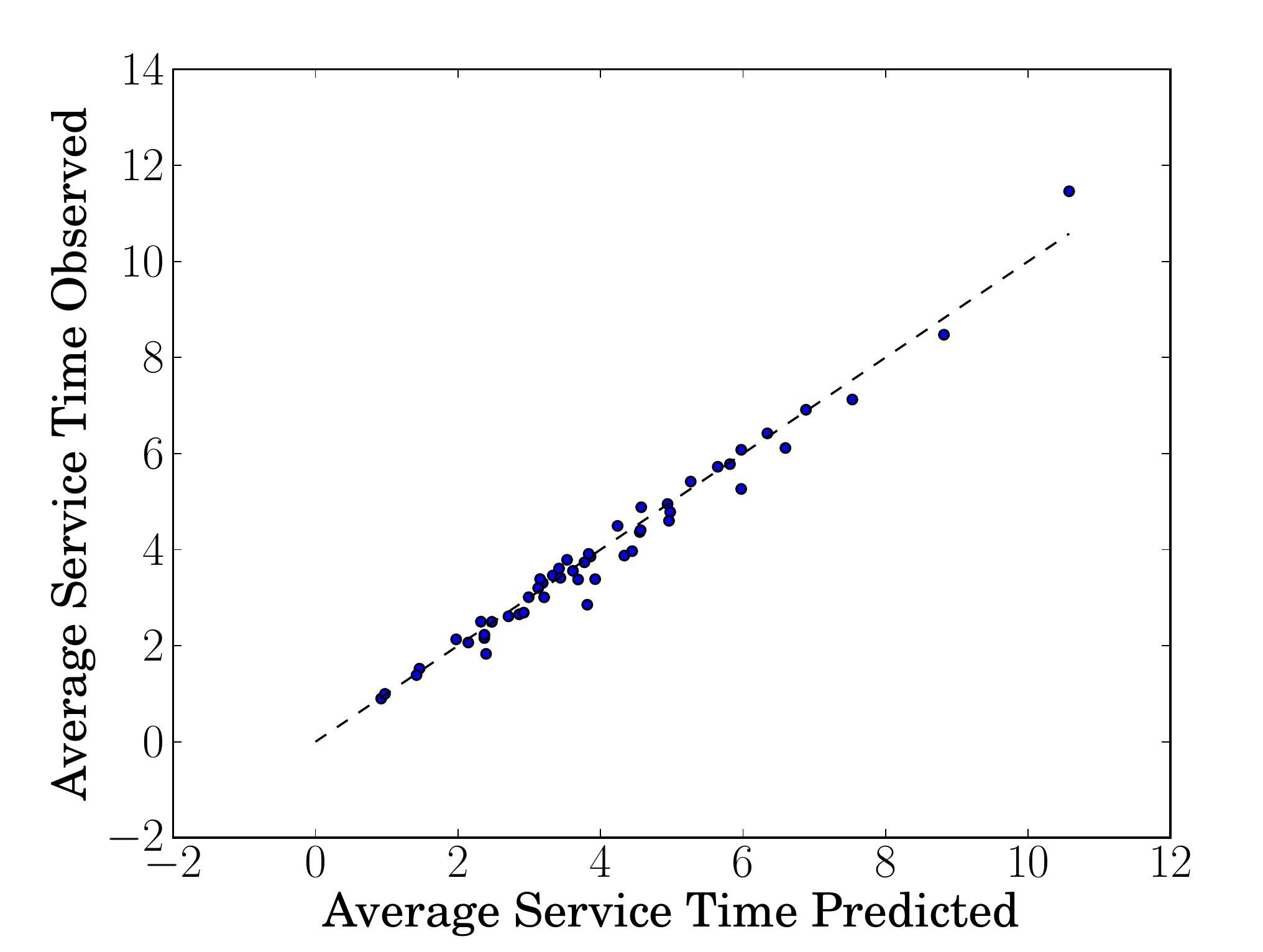}
\caption{Distribution of average service time (observed vs. predicted) over many random scenarios.}
\label{fig:service_time_results}
\end{figure}

In addition to the randomized scenarios,
we also considered again the road network of Section~\ref{sec:main result example},
with $\measvar$ defined by:
(i) with probability given by Table~\ref{tbl:simulation distribution}, $\randpt \in \roadvar_1$ and $\altrandpt \in \roadvar_2$;
(ii) given their road assignments, the coordinates of $\randpt$ and $\altrandpt$
are independent and uniformly distributed on each road interval.
\begin{table}
\centering
\begin{tabular}{c|c|cc|c}
									& 		& \multicolumn{2}{c}{ $\roadvar_2$ }		& $\measone$		\\ \hline
									& 		& E		& S			&					\\ \hline
	\multirow{2}{*}{$\roadvar_1$}	& N		& $1/5$	&			& $1/5$				\\
									& W		& $1/5$	& $3/5$		& $4/5$				\\ \hline
	$\meastwo$						&		& $2/5$	& $3/5$		&
\end{tabular}
\caption{
The probability mass function (pmf) $\measvar(\roadvar_1,\roadvar_2)$.
}
\label{tbl:simulation distribution}
\end{table}
The marginals of this distribution are equal to the input measures in Figure~\ref{fig:road example measures},
and so the EMD is equal to $31/30$;
the expected pickup-to-delivery distance
is equal to $17/15$, and
the sum of the terms is 
the predicted average per-demand service time 
$\overline s = 13/6 \approx 2.167$.

Figure~\ref{fig:number outstanding demands} shows plots
of the number of outstanding demands
over the duration $\simhorizon=10,000$
of two experiments with different arrival rates:
Figure~\ref{fig:number outstanding demands fast} shows the result
of the experiment with arrival rate $\arrivalrate = \arrivalrate^* + 0.1$.
The number of outstanding demands reaches
$\approx 1,000 = 0.1 \times \simhorizon$
by the final time,
showing strong corroboration of our predictions.
Figure~\ref{fig:number outstanding demands slow} shows the result
of the experiment with arrival rate $\arrivalrate = 0.99 \times \arrivalrate^*$,
which is below the stabilizable threshold.
The resulting plot includes several ``renewals'' (times when the system is empty)
and does not exhibit uncontrolled growth in the number of outstanding demands.
\begin{figure}
\hfill
\subfigure[$\arrivalrate = \arrivalrate^* + 0.1$]{
\includegraphics[width=.4\linewidth]{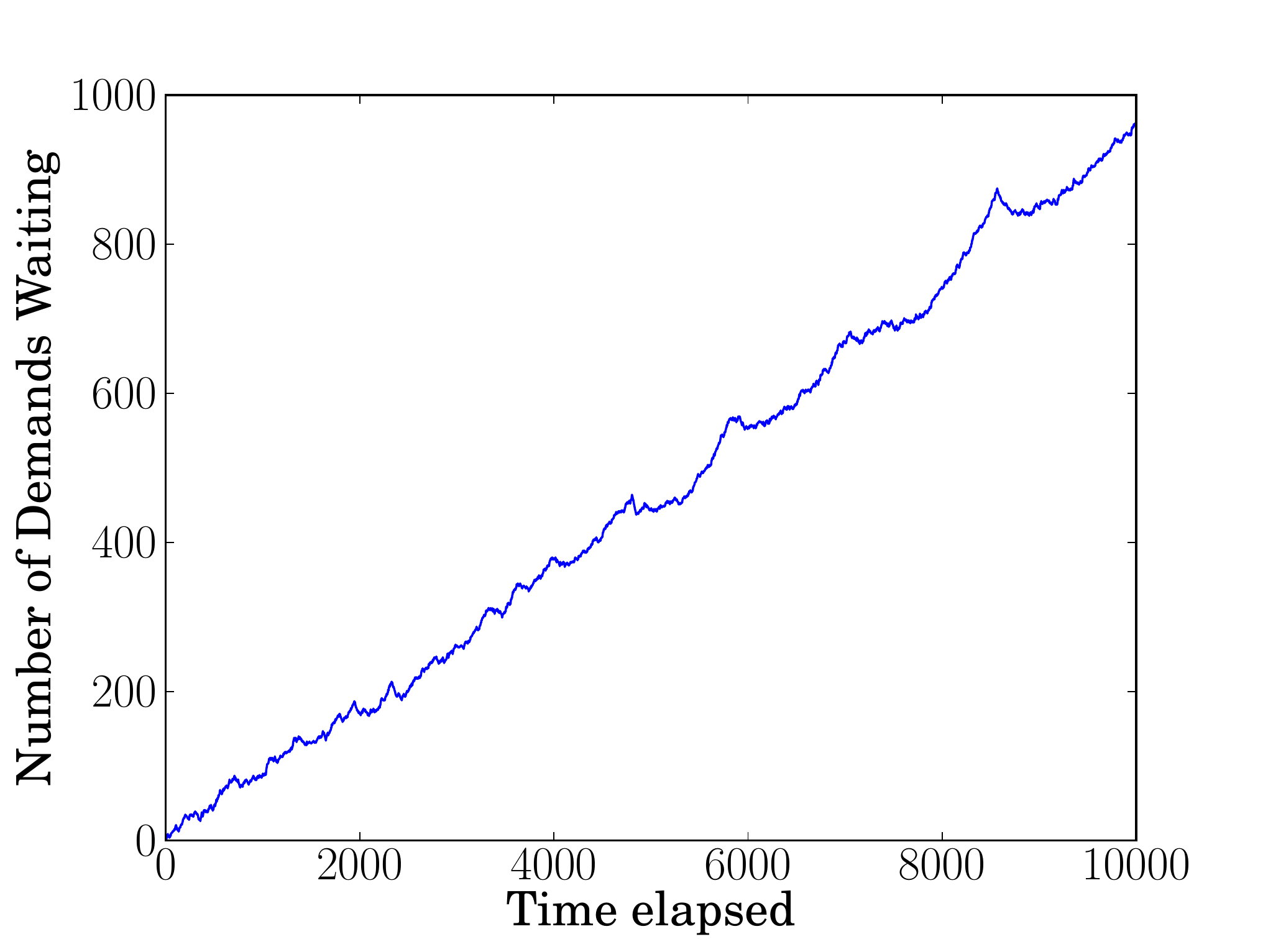}
\label{fig:number outstanding demands fast}
}
\hfill
\subfigure[$\arrivalrate = 0.99 \times \arrivalrate^*$]{
\includegraphics[width=.4\linewidth]{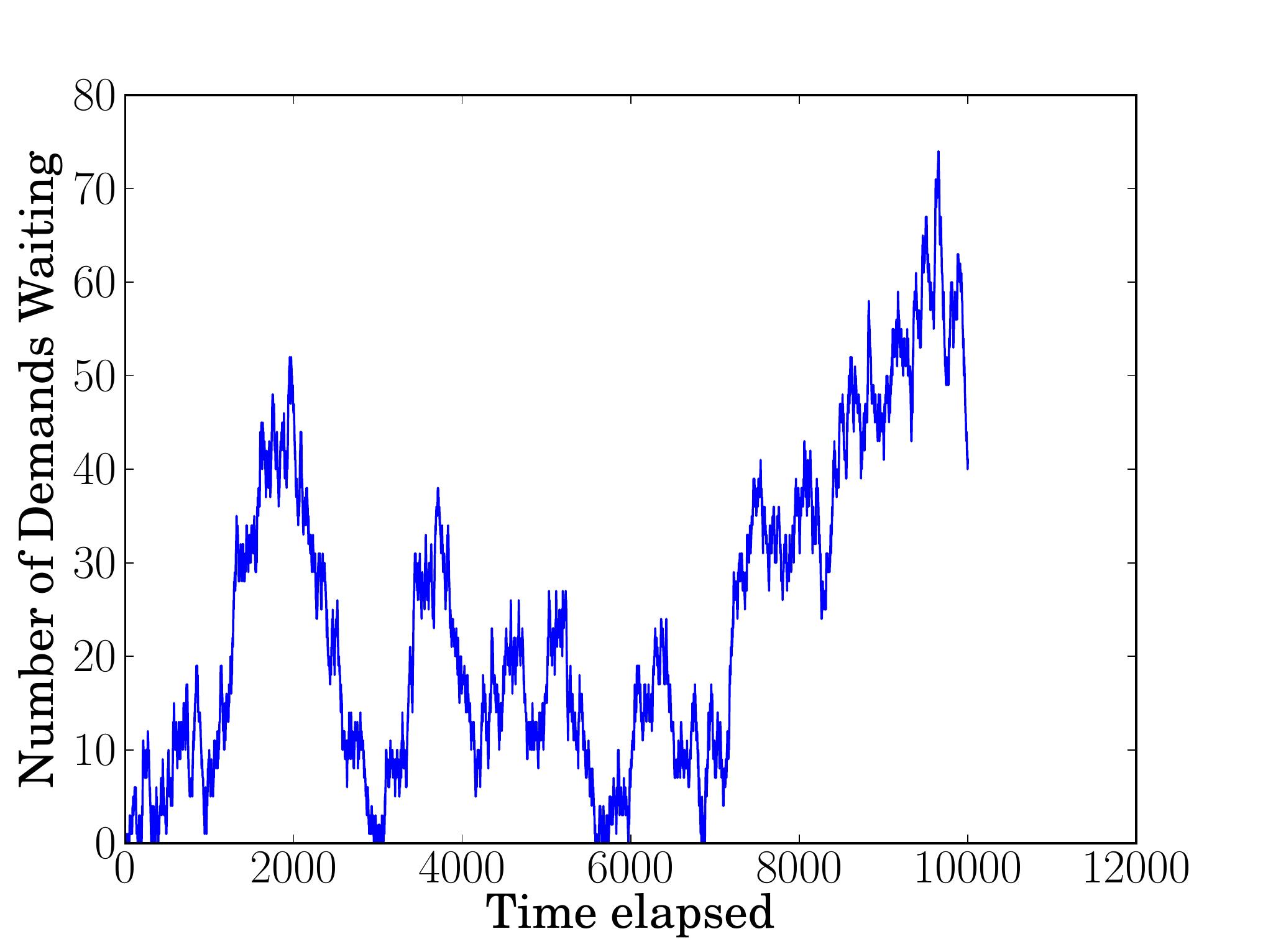}
\label{fig:number outstanding demands slow}
}
\hfill\null
\caption{
Plots of the number of outstanding demands as a function of time, over the duration of the simulation experiment.
}
\label{fig:number outstanding demands}
\end{figure}

\section{Approximating the Earth Movers Distance by Min-Cost Flow}
\label{sec:general purpose approx}

\newcommand{\numcell}{N}

\newcommand{\approxtag}{{\rm APPROX}}
\newcommand{\approxnet}{\flownet^\approxtag}
\newcommand{\newapproxnet}{{\tilde\flownet}^\approxtag}
\newcommand{\approxnodes}{\nodeset^\approxtag}
\newcommand{\approxedges}{\arcset^\approxtag}
\newcommand{\approxsupply}{\supplyvar^\approxtag}
\newcommand{\approxcosts}{\edgecosts^\approxtag}
\newcommand{\uppertag}{{\rm UPPER}}
\newcommand{\lowertag}{{\rm LOWER}}

\newcommand{\uppercosts}{\edgecosts^\uppertag}
\newcommand{\lowercosts}{\edgecosts^\lowertag}
\newcommand{\upperwts}{\edgewts^\uppertag}
\newcommand{\lowerwts}{\edgewts^\lowertag}
\newcommand{\upperwt}{\edgewt^\uppertag}
\newcommand{\lowerwt}{\edgewt^\lowertag}

\newcommand{\underdist}{\underline\distance}
\newcommand{\overdist}{\overline\distance}

\newcommand{\underapprox}{\Wass^\lowertag}
\newcommand{\overapprox}{\Wass^\uppertag}

\newcommand{\matchvar}{M}
\newcommand{\tmatch}{{\matchvar^\srcTag}}
\newcommand{\bmatch}{{\matchvar^\trgTag}}

The rest of the paper explores a particular method
to prove Theorem~\ref{thm:roadnet emd equation},
i.e., the correctness of our algorithm.
At a high-level, our approach is to develop an approximation scheme for $\Wass$ (the EMD),
bounding it entirely between an inner- and outer- approximation,
and then showing that the bounds converge (squeeze) to the LHS of~\eqref{eq:roadnet emd equation}.

\subsection{The General Purpose Scheme}
\label{sec:gp approx construction}

In this section
we present a naive, ``general-purpose''
approximation scheme for the Earth Movers distance
for a fairly general class of metric domains.
Specifically, 
we present a procedure which,
given a particular partition $\cellset$ and argument distributions $\measone$ and $\meastwo$,
generates a matched pair of network optimization problem instances.
The optimal solutions to these instances will bound $\Wass(\measone,\meastwo)$ from both sides.
If one can obtain a \emph{tessellation scheme} for the domain $\env$,
capable of tessellating any compact workspace $\workspace \subset \env$
to increasingly high ``resolution'',
then $\Wass$ can be approximated by making such bounds arbitrarily close.
(Such tessellation is easily obtainable, e.g., in Euclidean environments.)
%

%

\emph{Workspace Tesselation}:
The ability to tessellate is generally a property specific to the type of the domain $\env$.
A common tessellation scheme for Euclidean $\reals^d$
is the grid-based partition of $\reals^d$
into [hyper-] cubic cells of side-length $\onehalf \epsilon d^{-1/2}$.
The key objective of tessellation in this paper
is to ensure that for any $\epsilon > 0$ one can produce a partition $\cellset_\epsilon$ satisfying
\begin{equation}
\label{eq:partition property}
\begin{aligned}
	\max_{ \pointvar \in \workcell, \pointvar' \in \workcell' } \distance(\pointvar,\pointvar')
	-
	\min_{ \pointvar \in \workcell, \pointvar' \in \workcell' } \distance(\pointvar,\pointvar')
	\ \leq \
	\epsilon
	& \qquad \text{for all $(\workcell,\workcell') \in \cellset_\epsilon^2$}.
\end{aligned}
\end{equation}
%

\emph{Instance Construction}:
Let $\cellset$ be a finite partition of a workspace $\workspace \in \env$.
The flow network $\approxnet$ will comprise
a di-graph $(\approxnodes,\approxedges)$
and supplies $\approxsupply$.
We will call $\approxnet$ the \emph{approximation network}.
%
To construct the vertex set $\approxnodes$ we generate two sets $\tverts$ and $\bverts$ of new symbolic vertices;
each set is of cardinality $|\cellset|$.
We assign two such vertices to each cell $\workcell \in \cellset$,
one from the set $\tverts$ and one from the set $\bverts$, where
each vertex is assigned to a single cell only (see Figure~\ref{fig:partition and assignment}).
%
\begin{figure}[h!]
\centering
\begin{tikzpicture}
\tikzstyle{groupstyle}=[fill,blue,opacity=.1]
\tikzstyle{gridstyle}=[draw,very thick]

\newcommand{\rectThreeD}[3]		
{
\path[#3] (-4+#1,-.5+#2) -- (#1,-.5+#2) -- (4+#1,.5+#2) -- (#1,.5+#2) -- cycle ;
}
\newcommand{\latticeThreeD}[3]	
{
\path[#3]
	(-2+#1,-.5+#2) -- (2+#1,.5+#2)
	(-2+#1,#2) -- (2+#1,#2) ;
}

\rectThreeD{0}{1.2}{groupstyle}
\rectThreeD{0}{0}{gridstyle}
\latticeThreeD{0}{0}{gridstyle}
\rectThreeD{0}{-1.1}{groupstyle,red}

\clip (-4.2,-2) rectangle (4.75,2.5) ;	
\node at (4.2,.5) {$\workspace$} ;
\node at (1,.75) {$\workcell^1$} ;
\node at (3,.75) {$\workcell^2$} ;
\node at (-1,-.65) {$\workcell^3$} ;
\node at (-3,-.65) {$\workcell^4$} ;
\node at (4.2,1.7) {$\tverts$} ;
\node at (4.2,-.55) {$\bverts$} ;
\foreach \x/\y/\name in
{ -.25/1.5/$u^1$, 1.5/1.4/$u^2$, .75/1.1/$u^3$, -1.3/1.1/$u^4$  }
\draw (\x,\y)
	[fill] circle (2pt)
	node [above] {\name}
	-- ++(0,-1.25) node {$\times$} ;
\foreach \x/\y/\name in		
{ 0/1.5/$v^1$, 2.1/1.5/$v^2$, .25/1/$v^3$, -2.1/.9/$v^4$ }
\draw (\x,\y-2.35)
	[fill] circle (2pt)
	node [below] {\name}
	-- ++(0,1.25) node {$\times$} ;
\end{tikzpicture}
\caption{
Bipartite assignment of symbolic vertices ($\tverts$ and $\bverts$) to the cells in $\cellset$.
}
\label{fig:partition and assignment}
\end{figure}
%
Let bipartite matchings $\tmatch$ (between $\tverts$ and $\cellset$) and $\bmatch$ (between $\bverts$ and $\cellset$) 
denote the respective assignments.
(For example, if $u$ is the vertex in $\tverts$ assigned to $\workcell \in \cellset$,
then $(u,\workcell) \in \tmatch$.)
We define the supplies as
\begin{align}
\approxsupply( u ) := \measone( \workcell )
     && ( (u,\workcell) \in \tmatch ),
\label{eq:supplies supply}				\\
 \approxsupply( v ) := -\meastwo(\workcell)
     && ( (v,\workcell) \in \bmatch ).
\label{eq:supplies demand}
\end{align}
%
%
Let $\approxedges$ form the complete bipartite graph between $\tverts$ and $\bverts$,
i.e.,
$\approxedges :=$
$\tverts \times \bverts$.
%
%
Let $\lowerwts =: \{ \underline\edgewt_{(u,v)} \}$
be set the set of edge weights on $\approxedges$ satisfying
\begin{equation} \label{eq:defn lower weights}
  \underline\edgewt_{( u, v )} =
  	\min_{ \pointvar \in \workcell, \pointvar' \in \workcell' }
    \distfunc{\pointvar,\pointvar'}
  \qquad
  \text{for $(u,\workcell) \in \tmatch$, $(v,\workcell') \in \bmatch$},
\end{equation}
and let $\upperwts =: \{ \overline\edgewt_{(u,v)} \}$ be the set of edge weights satisfying
\begin{equation}
  \overline\edgewt_{(u,v)} =
  \max_{ \pointvar \in \workcell^i, \pointvar' \in \workcell^j }
    \distfunc{\pointvar,\pointvar'}
  \qquad
  \text{for $(u,\workcell) \in \tmatch$, $(v,\workcell') \in \bmatch$}.
\end{equation}

\subsection{Approximation Bounds}
\label{sec:gp approx bounds}

The network $\approxnet$ captures a hypothetical scenario
(by aggregation of points into a finite number of cells)
where the cost of transportation (distance) from one cell to another is a single constant 
regardless of the particular endpoints.
The costs $\lowercosts$ are ``optimistic'', assigning cost to a pair of cells
equal to the minimum distance between endpoints in either cell;
the costs $\uppercosts$, meanwhile, are ``pessimistic'', assigning cost equal to the maximum such distance.
As the fine-ness of the tesselation increases, in the sense that $\epsilon \to 0^+$ in~\eqref{eq:partition property},
the difference between the optimistic and pessimistic costs will vanish.
%
Such intuition supports the claims of Propositions~\ref{prop:approx_contains_wass} and~\ref{prop:approx_squeezes}, below;
the formal proofs, however, are provided in Appendix~\ref{sec:naive approx correctness}.
\begin{prop}
\label{prop:approx_contains_wass}
For any distributions $\measone$ and $\meastwo$
satisfying Assumptions~\ref{assump:nicemeasures} and \ref{assump:road_posmutex},
any $\epsilon > 0$, and
any partition $\cellset_\epsilon$ of workspace $\workspace \subset \env$
satisfying~\eqref{eq:partition property},
let $\approxnet$ denote the approximation network of Section~\ref{sec:gp approx construction}
having weights $\lowerwts$ and $\upperwts$.
Let
\begin{align}
\underapprox &\doteq
	\min_{ \flowmap \in \approxnet } \flowcost(\flowmap ; \lowerwts )	
\label{eq:defn underapprox}		\\
\overapprox &\doteq
	\min_{ \flowmap \in \approxnet } \flowcost(\flowmap ; \upperwts ).
\label{eq:defn overapprox}
\end{align}
Then
$\underapprox \leq \Wass(\measone,\meastwo) \leq \overapprox$.
\end{prop}
\begin{prop}
\label{prop:approx_squeezes}
Under the same condition as Proposition~\ref{prop:approx_contains_wass},
$\overapprox - \underapprox \leq \epsilon |\measvar|$,
where $|\measvar|$ denotes the constant total measure of either $\measone$ or $\meastwo$.
\end{prop}
Together, Propositions~\ref{prop:approx_contains_wass} and~\ref{prop:approx_squeezes} prove that
$\underapprox \to \Wass(\measone,\meastwo)^-$ and
$\overapprox \to \Wass(\measone,\meastwo)^+$
as $\epsilon \to 0^+$,
i.e., both converge to $\Wass(\measone,\meastwo)$.
%

\section{Approximating the EMD on Road Networks}
\label{sec:roadnet approx}

\subsection{The General-Purpose Scheme}

Road networks are sufficiently like Euclidean $\reals^1$
that a small modification to the grid-based tessellation scheme of Section~\ref{sec:gp approx construction}
obtains the same convergence in the approximation by $\approxnet$
as the grid-based scheme does for $\reals^d$:
%
For each $\roadvar \in \roadset$,
let $N_\roadvar := \lceil \roadlen_\roadvar / \epsilon \rceil$
and let $\epsilon_\roadvar := \roadlen_\roadvar / \numcell_\roadvar$.
Then one can partition each road $\roadvar \in \roadset$ into $N_\roadvar$ segments of length $\epsilon_\roadvar$.
We will refer to such partition as the \emph{$\epsilon$-tesselation} of $\roadnet$.
The interval lengths $\{ \epsilon_\roadvar \}_{\roadvar\in\roadset}$ are all smaller than $\epsilon$,
so the resulting partition satisfies~\eqref{eq:partition property} and
Propositions~\ref{prop:approx_contains_wass} and~\ref{prop:approx_squeezes} hold.

%

While our pain-staking attention to network flow-based approximation schemes may be mis-leadingly algorithmic,
our interest in them is
\emph{not} to approximate $\Wass$, but
to discover a sequence
$\Wass_k$ which converges to $\Wass$ and has an analytical limit.
%
Unfortunately, the network structure generated by the general-purpose scheme is too general
to reveal any 
underlying analytical form of $\Wass$.
Fortunately, that scheme is \emph{not} the only network flow-based approximation scheme that we may use.

\subsection{The Path-based Scheme}
\label{sec:roadnet approx construction}

\newcommand{\pathlen}[1]{\left| #1 \right|}

\newcommand{\pathtag}{{\rm PATH}}

\newcommand{\pathnet}{\flownet^\pathtag}
\newcommand{\pathnodes}{\nodeset^\pathtag}
\newcommand{\pathedges}{\arcset^\pathtag}
\newcommand{\pathcosts}{\edgecosts^\pathtag}
\newcommand{\pathwts}{\edgewts^\pathtag}
\newcommand{\pathsupply}{\supplyvar^\pathtag}

\newcommand{\roadnode}[1][k]{ {\nodevar_\roadvar^{#1}} }
\newcommand{\roadcell}[1][k]{ {\workcell_\roadvar^{#1}} }

\newcommand{\roadcost}[1][\roadvar]{ \flowcost_{#1} }

In this section, we present another approximation scheme which leverages the structure of the road network $\roadnet$.
We will call our alternative approximation scheme the ``path-based'' scheme.
An important feature of the scheme is that it uses
the same $\epsilon$-tesselation of $\roadnet$, and
many of the same network vertices (i.e., $\approxnodes$), as the general-purpose one.
The scheme differs in that we seek an alternative flow network topology.
Our goal is to obtain additional insight into computing the EMD.
Naturally, the new scheme must preserve the cost of the min-cost flow.
(Because either of the squeezing bounds converges to $\Wass$,
we focus only on the lower bound produced by $\lowercosts$.)


The ability to produce a meaningful alternative topology is based on two important observations about network flows:
First, while network flows are most commonly represented as mappings from individual edges to flow volume,
they can be represented equally well by mapping from \emph{paths} to flow volume.
%
For example, the network flow in Figure~\ref{fig:roadnet example solution flow}
can be interpreted as a so-called ``path and cycle flow'',
with
$1/5$ unit flow on the path (E-2-N), 
$2/15$ flow on the path (E-2-1-W),
$1/15$ flow on the path (E-3-4-W), and
$3/5$ flow on the path (S-4-W).
The second observation is that in the absense of edge ``capacities'' (which do not arise in this paper),
minimum-cost network flows are supported \emph{entirely} on shortest paths.

\begin{definition}[Path and cycle flows]
Let $\pathset$ denote the set of simple paths on a (multi-)digraph $G=(\nodeset,\edgeset)$,
and let $\cycleset$ denote the set of cycles.
A \emph{path and cycle flow} is a mapping $\flowmap : \pathset \cup \cycleset \to \reals_{\geq 0}$.
(We will call flows of the former type ($\arcset \to \reals_{\geq 0}$) \emph{arc flows}, or simply flows.)
\end{definition}
Path and cycle flows determine arc flows in a natural way,
such that the flow on an edge is equal to the sum of all flows on paths and cycles that use the edge.
%
Defining the delta function $\delta_\arcvar(\pathvar)$ for each $\arcvar\in\arcset$---%
equal to $1$ if $\arcvar$ is included in the path or cycle $\pathvar \in \pathset \cup \cycleset$, and $0$ otherwise---%
then the arc flow $\newflowmap$ described by a path and cycle flow $f$ is determined by
\begin{equation} \label{eq:arc_flow_representation}
  \newflowmap(\arcvar)
  =
  \sum_{\pathvar \in \pathset \cup \cycleset} \delta_\arcvar(\pathvar) \flowmap(\pathvar)
  \qquad \text{for all $\arcvar\in\arcset$}.
\end{equation}
A path and cycle flow is admissible if its arc flow is admissible.
Letting $\pathlen{ \pathvar }_\edgewts$
denote the total weight of a path $\pathvar$
on a weighted network $(\flownet,\edgewts)$, i.e.,
$\pathlen{\pathvar}_\edgewts \doteq \sum_{\arcvar \in \arcset} \delta_\arcvar(\pathvar) \edgewts(\arcvar)$,
the cost of a path-and-cycle flow can be written
$\flowcost( \flowmap ; \edgewts)
	\doteq \sum_{ \pathvar \in \pathset \cup \cycleset } \flowmap(\pathvar) \pathlen{ \pathvar }_\edgewts$.
A path-and-cycle flow has the same total weight as its arc flow.

\begin{lemma} \label{lemma:equivalent networks}
Let $(\flownet,\edgewts)$ and $(\newflownet,\newedgewts)$ be two weighted flow networks satisfying the following properties:
\begin{enumerate}
\item
\label{property:same_supplies}
Every supply vertex has the same supply in $\flownet$ and $\newflownet$;
\item 
\label{property:same_demands}
Every demand vertex has the same demand in $\flownet$ and $\newflownet$;
\item \label{property:shortest_paths_equal}
The total weight of the weighted shortest path,
from any supply vertex to any demand vertex,
is the same in both networks.
\end{enumerate}
Let $\flowcost^*$ and $\newflowcost^*$
denote the costs of the minimum-cost flows on $\flownet$ and $\newflownet$,
respectively (and with respective weights).
Then $\flowcost^*$ and $\newflowcost^*$ are equal.
\end{lemma}

By Lemma~\ref{lemma:equivalent networks}, it is possible to substitute an alternative topology
over the network vertices $\approxnodes$,
\emph{without} changing the value of the minimum cost flow,
so long as every \emph{shortest path} from a supply vertex $u$ to a demand vertex $v$ has length
equal to the weight of edge $(u,v)$ in $\approxnet$.
Our proof of the lemma requires elements of the next Theorem, reproduced from~\cite{ahuja1993network}:
\begin{theorem}[Theorem~3.5 of~\cite{ahuja1993network} (annotated)]
\label{thm:path and cycle flows}
Every path and cycle flow has a unique representation as nonnegative arc flows [i.e., \eqref{eq:arc_flow_representation}].
Conversely, every nonnegative arc flow can be represented as a path and cycle flow
(though not necessarily uniquely) with the following two properties:
\begin{enumerate}
\item \label{itm:supply_to_demand}
Every directed path with positive flow connects a [supply] node to [a demand] node.
\item (not needed for our discussion, see~\cite{ahuja1993network} for full text).
\end{enumerate}
\end{theorem}

\begin{proof}[Proof of Lemma~\ref{lemma:equivalent networks}]
It is sufficient to prove $\newflowcost^* \leq \flowcost^*$, since the two networks commute in the statement of the lemma.
Let $\flowmap^*$ be the path-and-cycle representation of the minimum-cost flow on $\flownet$.
By Property~1 of Theorem~\ref{thm:path and cycle flows}, every positive-flow path is from a supply node to a demand node.
Each positive-flow path is also a shortest path (this can be proved by a simple substitution argument).
We can construct a path-and-cycle flow $\newflowmap$ on $\newflownet$ by adding the weight of each positive-flow path in $\flowmap^*$
into $\newflowmap$ on the shortest directed path between the same endpoints.
Properties~\ref{property:same_supplies} and~\ref{property:same_demands} of Lemma~\ref{lemma:equivalent networks}
ensure that $\newflowmap \in \newflownet$ (it is admissible).
By Property~\ref{property:shortest_paths_equal}, the latter paths have the same weight as the former ones,
proving the total cost of $\newflowmap$ is the same as that of $\flowmap^*$.
$\newflowcost^*$, by definition, cannot be more.
\end{proof}

\emph{Instance Construction}:
\newcommand{\subgraph}{g}
\newcommand{\roaddevice}[1][\roadvar]{\subgraph_{#1}}
Our construction must satisfy Lemma~\ref{lemma:equivalent networks} with $\approxnet$.
Note that Properties~\ref{property:same_supplies} and~\ref{property:same_demands} are quite easy to satisfy,
i.e., by letting $\pathsupply$ equal $\approxsupply$ on $\supplyset \cup \demandset$ and
zero anywhere else.
%
In order to satisfy Property~\ref{property:shortest_paths_equal},
we seek to construct a network
where the shortest path
from $u \in \tverts$ ($(u,\workcell) \in \tmatch$)
to $v \in \bverts$ ($(v,\workcell') \in \bmatch$)
has total weight equal to that given by $\lowerwts$, or
%
the minimum distance on $\roadnet$ from $\workcell$ to $\workcell'$, i.e.~\eqref{eq:defn lower weights}.
The crucial observation is that any path from $\workcell$ to $\workcell'$ can be decomposed into three parts:
(i) first, a path from $\workcell$ to an endpoint $\roadvar^\pm$
of the road $\roadvar \in \roadset$ for which $\workcell \subset \roadvar$;
(ii) second, a path from that endpoint $\roadvar^\pm$ to an endpoint $\tilde\roadvar^\pm$ of another road $\tilde\roadvar \in \roadset$, $\workcell' \subset \tilde\roadvar$;
(iii) finally, a path from the second endpoint $\tilde\roadvar^\pm$ to the cell $\workcell'$.

To obtain the network $\pathnet$ instance we start with
$\pathnodes := \roadverts$ (the vertices of $\roadnet$) and
$\pathedges := \emptyset$.
Then, for each \emph{non-transshipment} road $\roadvar \in \supplyset \cup \demandset$,
we insert into the graph $(\pathnodes,\pathedges)$ one of two possible ``\emph{road devices}''.
If $\roadvar$ is a supply road, i.e., $\roadvar \in \supplyset$,
then we add a ``supply device'', as shown in Figure~\ref{fig:supply device};
\begin{figure}
\centering
\providecommand{\numcell}{N}		

\begin{tikzpicture}
\tikzstyle{roadstuff}=[black!25]
\tikzstyle{graphstuff}=[black]
\tikzstyle{graphedge}=[graphstuff,->,very thick]
\tikzstyle{graphloop}=[graphedge,bend right]

\tikzstyle{endconn}=[circle,dashed,inner sep=5,draw]

\draw [roadstuff,line width=2pt]
	(-4,0) -- (-.8,0)		
	(.8,0) -- (4,0) ;		
\foreach \k in { -4,...,-1, 1,2,...,4 }
	\draw [roadstuff] (\k,-.2) -- (\k,.2) ;
\foreach \x / \label in {
	-3.5 / $\workcell^1$,
	-2.5/$\workcell^2$,
	-1.5/$\workcell^3$,
	1.5/$\workcell^{\numcell-2}$,
	2.5/$\workcell^{\numcell-1}$,
	3.5/$\workcell^{\numcell}$
} \draw[] (\x,-.5) node {\label} ;

\draw[decorate,decoration=snake]
	(-.5,-.5) -- (-.5,3.5)
	(.5,-.5) -- (.5,3.5) ;

\newcommand{\nodelevel}{2}
\newcommand{\drawNode}[3]{		
\draw[graphstuff]
	(#2,\nodelevel)
	node [shape=circle,fill,inner sep=2] (#1) {}
	node [above of=#1] {#3} ;
\draw[roadstuff,dashed]
	(#1.south) -- (#2,0) node {$\times$} ;
}
\newcommand{\drawLoop}[2]{
\draw[graphloop,->,bend right]
	(#1) edge node [below] {$\roadlen/\numcell$} (#2)
	(#2) edge node [above] {$\roadlen/\numcell$} (#1)  ;
}

\drawNode{node-1}{-3.5}{ $u^1$ }
\drawNode{node-2}{-2.5}{ $u^2$ }
\drawNode{node-3}{-1.5}{ $u^3$ }
\node (node-4) at (-.5,\nodelevel) {} ;
\node (node-n-minus-3) at (.5,\nodelevel) {} ;
\drawNode{node-n-minus-2}{1.5}{ $u^{\numcell-2}$ }
\drawNode{node-n-minus-1}{2.5}{ $u^{\numcell-1}$ }
\drawNode{node-n}{3.5}{ $u^\numcell$ }

\drawLoop{node-1}{node-2}
\drawLoop{node-2}{node-3}
\drawLoop{node-n-minus-2}{node-n-minus-1}
\drawLoop{node-n-minus-1}{node-n}

\begin{scope}
\clip (-3,1) rectangle (-1,3) (1,1) rectangle (3,3) ;
\draw[graphloop]
	(node-3) edge (node-4)
	(node-4) edge (node-3) ;
\draw[graphloop]
	(node-n-minus-2) edge (node-n-minus-3)
	(node-n-minus-3) edge (node-n-minus-2) ;
\end{scope}

\draw (-4,1) [dashed] rectangle (4,3.5) node [right] {$\roadvar \in \supplyset$} ;

\draw[graphstuff]
	(-5.5,2) node[endconn] (left-end) {$\roadvar^-$} ;
\draw[graphedge]
	(node-1) edge node[above] {$\roadconnleft$} (left-end) ;
\draw[graphstuff]
	(5.5,2) node[endconn] (right-end) {$\roadvar^+$} ;
\draw[graphedge]
	(node-n) edge node[above] {$\roadconnright$} (right-end) ;

\end{tikzpicture}
\caption{
The device $\roaddevice$ of a supply road $\roadvar \in \supplyset$.
}
\label{fig:supply device}
\end{figure}
The vertices of this device are the ones in $\tverts \subset \approxnodes$
associated with the tessellation of $\roadvar$;
as seen in Figure~\ref{fig:supply device},
they are ordered from $\roadvar^-$ to $\roadvar^+$.
Otherwise, if $\roadvar$ is a demand road ($\roadvar \in \demandset$), then
we add a ``demand device'', which
is like the supply device, except
(i) the vertices are those from $\bverts$, and
(ii) $\roadconnleft$ and $\roadconnright$ have the opposite direction.
(In either case, $\roadconnleft$ has endpoints $\roadnode[1]$ and $\roadvar^-$, and
$\roadconnright$ has endpoints $\roadnode[\numcell]$ and $\roadvar^+$.)
We denote by $\roaddevice$ the device subgraph belonging to road $\roadvar$.
%
%
\begin{remark}
The resulting set $\pathnodes$ is not exactly that same set as $\approxnodes$.
We observe, however, that
the symmetric difference set includes only non-supply, non-demand vertices, which 
cannot contribute positive flow paths to a minimum-cost flow;
thus, they do not affect compliance with Lemma~\ref{lemma:equivalent networks}.
\end{remark}

As indicated in Figure~\ref{fig:supply device},
let the weights $\pathwts$ give
$\epsilon_\roadvar = \roadlen_\roadvar / n$ on all the road device edges
except $\roadconnleft$ and $\roadconnright$ which are ``free'' (zero cost).
Such weights are carefully chosen to ensure that:
(i) the shortest path from $u \in \roaddevice$ to either endpoint $\roadvar^\pm$
has total weight equal to the distance on $\roadnet$ from $\workcell$ to $\roadvar^\pm$;
(ii) the shortest path from either endpoint $\tilde\roadvar^\pm$ to $v \in \roaddevice[\tilde\roadvar]$
has total weight equal to the distance from $\tilde\roadvar^\pm$ to $\workcell'$.
Finally, we insert into $\pathedges$ the set of routing edges $\routingedges$ from Section~\ref{sec:main result construction},
with weights $\routingwts$.
These weights are chosen so that the shortest path on $\routingedges$
from $i \in \roadverts$ to $j \in \roadverts$
has total weight equal to $\distance(i,j)$.

%
%
%
%
%

\begin{prop} \label{prop:pathnet similar to approxnet}
For any road network $\roadnet$,
argument distributions $\measone$ and $\meastwo$
satisfying Assumptions~\ref{assump:nicemeasures} and~\ref{assump:road_posmutex},
and $\epsilon > 0$,
let $\cellset_\epsilon$ denote the $\epsilon$-tessellation of $\roadnet$,
let $\approxnet$ denote the Wasserstein network generated by Section~\ref{sec:gp approx construction},
with weights $\lowerwts$, and
let $\pathnet$ denote the network generated by Section~\ref{sec:roadnet approx construction}
with weights $\pathwts$.
$(\approxnet,\lowerwts)$ and $(\pathnet,\pathwts)$
are equivalent in the sense of Lemma~\ref{lemma:equivalent networks}.
\end{prop}
The reasoning behind the proposition is the same as that of the construction.
We omit the redundant formal proof.

%

Combining Proposition~\ref{prop:pathnet similar to approxnet} and Lemma~\ref{lemma:equivalent networks}
shows that
$\min_{\flowmap \in \pathnet }
	\flowcost( \flowmap ; \pathwts ) = \underapprox$,
and so proves its convergence to $\Wass(\measone,\meastwo)$ from below as $\epsilon \to 0^+$.

%


%
%
%
%
%

\section{Analysis of Exact and Approximation Algorithms}
\label{sec:algorithm analysis}

In this section we analyze the complexity of construction of
the three networks $\wassnet$, $\approxnet$, and $\pathnet$.
In particular, we consider the way that the sizes of the instance graphs relate to
(i) the size of the road network $\roadnet$
(both the size of its graphical representation and its \emph{physical} size as determined by the lengths of roads); and
(ii) the fine-ness $\epsilon$ of the input tessellation (in the case of approximation).
Finally, we present a numerical study of graph sizes, approximation quality,
and the runtime of a standard QP-based algorithm to compute each solution
for the example network of Figure~\ref{fig:road example measures}.

\subsection{Complexity}
\providecommand{\lcard}{\left|}
\providecommand{\rcard}{\right|}
\providecommand{\cardgroup}[1]{\lcard #1 \rcard}

The remarkable feature of $\wassnet$ is that it depends only on
the size of the \emph{representation} of $\roadnet$, and not on its \emph{physical} size.
$\wassnodes$ has size equal to $\cardgroup{\roadverts} + \cardgroup{\roadset}$,
and $\wassedges$ has size bounded by $4 \cardgroup{\roadset}$;
there are exactly two decision edges and as many as two routing edges per road $\roadvar \in \roadset$.
The size of $\approxnet$, on the other hand,
depends on the physical size of the network and on the approximation parameter $\epsilon$.
$\approxnodes$ has size equal to $2\cardgroup{\cellset_\epsilon}$
or $2 \sum_{\roadvar \in \roadset} \numcell_\roadvar$, which goes as $\Theta(1/\epsilon)$.
$\approxedges$ has size equal to $\cardgroup{\cellset_\epsilon}^2$, which has dominating complexity $\Theta(1/\epsilon^2)$.
Note that such growth of $\approxnet$ may be quite impractical to approximate the EMD with realistic road networks
with hundreds or even thousands of miles of streets.
$\pathnet$ leverages the structure of the road network to reduce the space complexity of approximation to $\Theta(1/\epsilon)$.
$\pathnodes$ has size equal to $\cardgroup{\roadverts} + \cardgroup{\cellset_\epsilon}$
and $\pathedges$ has size bounded by $2\cardgroup{\roadset} + 2\cardgroup{\cellset_\epsilon}$.
Note that the size of $\pathnet$ depends on \emph{both} the physical size of the road network and the size of its representation.

%
%

\subsection{Numerical Study}

Figure~\ref{fig:numerical vertices} shows a plot of the number of vertices instantiated in
$\wassnet$, $\approxnet$, and $\pathnet$, as a function of $\epsilon$, for
the EMD problem discussed in Section~\ref{sec:main result example} (Figure~\ref{fig:road example measures}).
Figure~\ref{fig:numerical edges} shows a plot of the number of edges instantiated.
$\wassnet$ exhibits a flat response to $\epsilon$ in both plots, since it does not depend on the parameter.
As expected, both approximation schemes exhibit the same rate of growth ($\Omega(1/\epsilon)$) in the number of vertices instantiated,
while $\approxnet$ has a factor $1/\epsilon$ greater growth in the rate of edges instantiated.
\begin{figure}
\hfill
\subfigure[
The number of vertices instantiated.
]{
\includegraphics[width=.45\linewidth]{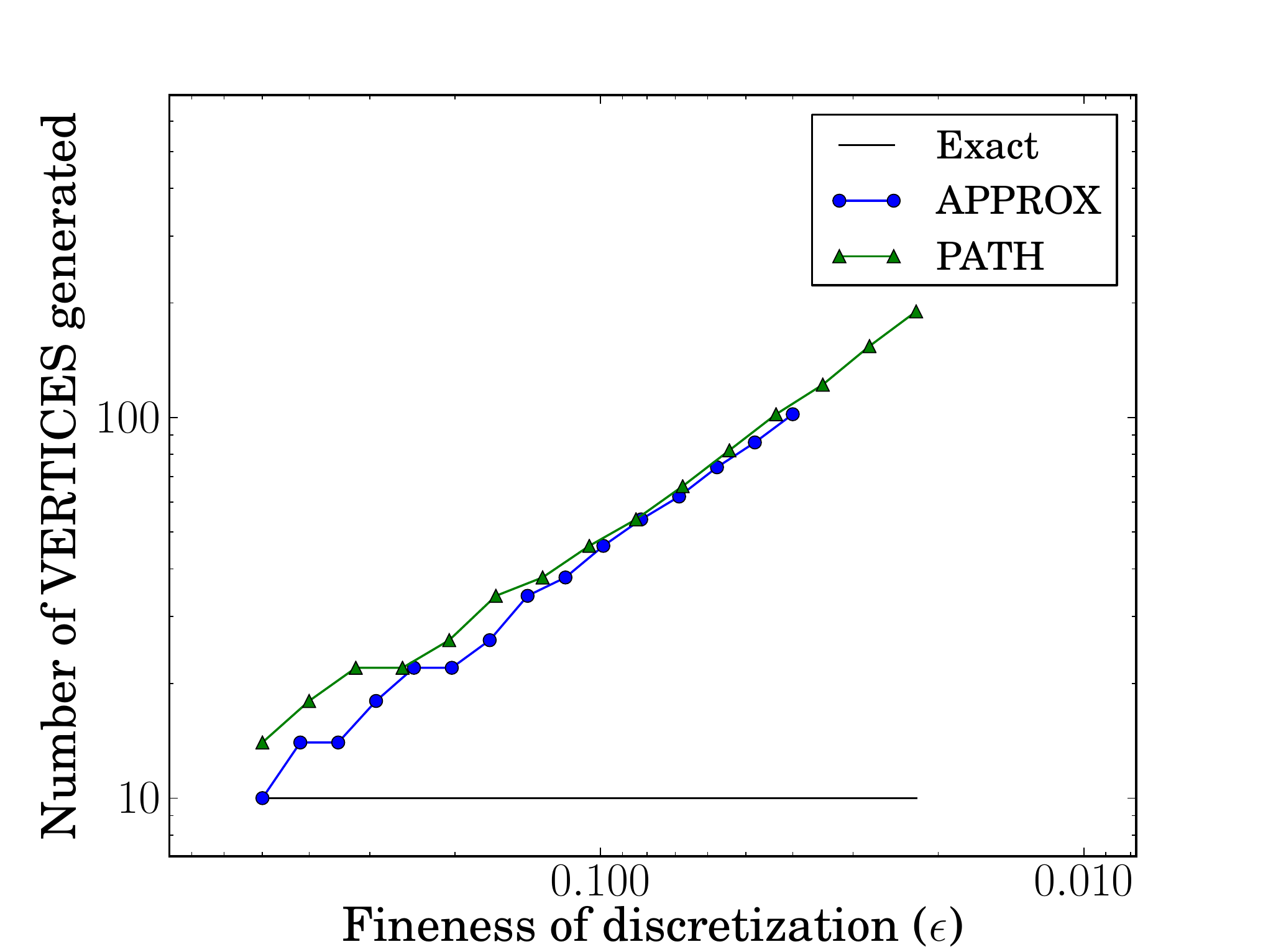}
\label{fig:numerical vertices}
}
\hfill
\subfigure[
The number of edges instantiated.
]{
\includegraphics[width=.45\linewidth]{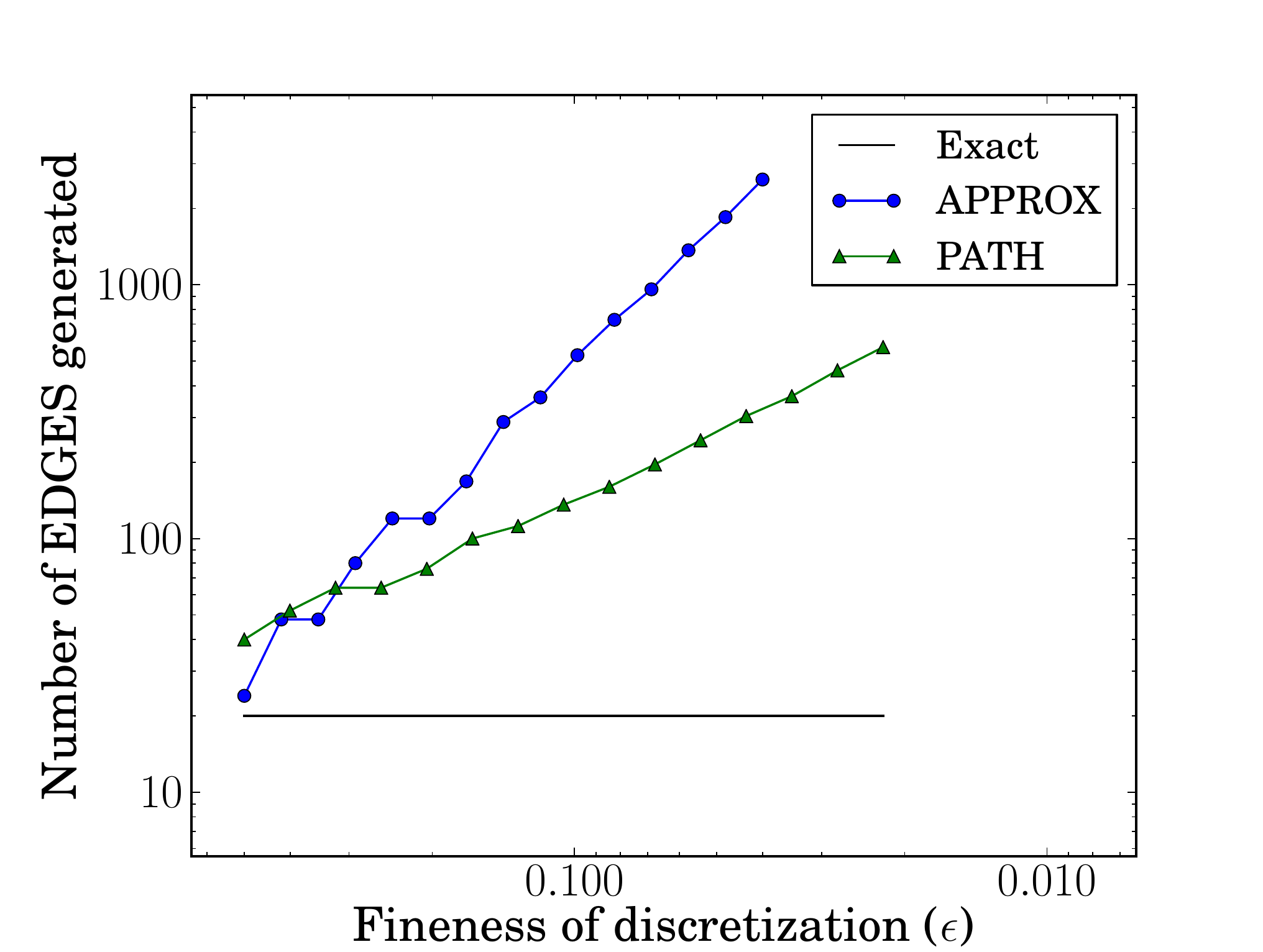}
\label{fig:numerical edges}
}
\hfill\null
\caption{
Number of objects instantiated, in the Wasserstein network,
as a function of the fine-ness parameter $\epsilon$
of the $\epsilon$-tessellation of the roadmap in Figure~\ref{fig:road example measures}.
}
\end{figure}

Figure~\ref{fig:numerical closeness} shows a plot of the quality of approximation
of the methods in Sections~\ref{sec:gp approx construction} and~\ref{sec:roadnet approx construction},
respectively, for values of the resolution parameter $\epsilon$ as small as possible under space and runtime considerations
(e.g., producing less than $100,000$ graph objects, and running in minutes on an Intel i5 processor with 4 CPUs and 4GB of RAM).
The dashed center line marks the solution obtained by the exact algorithm, i.e., optimization over the flow network in Figure~\ref{fig:roadnet example solution}.
The plot shows convergence of the approximation bounds to the value predicted by $\wassnet$.
\begin{figure}
\hfill
\subfigure[
EMD approximation bounds.
]{
\includegraphics[width=.45\linewidth]{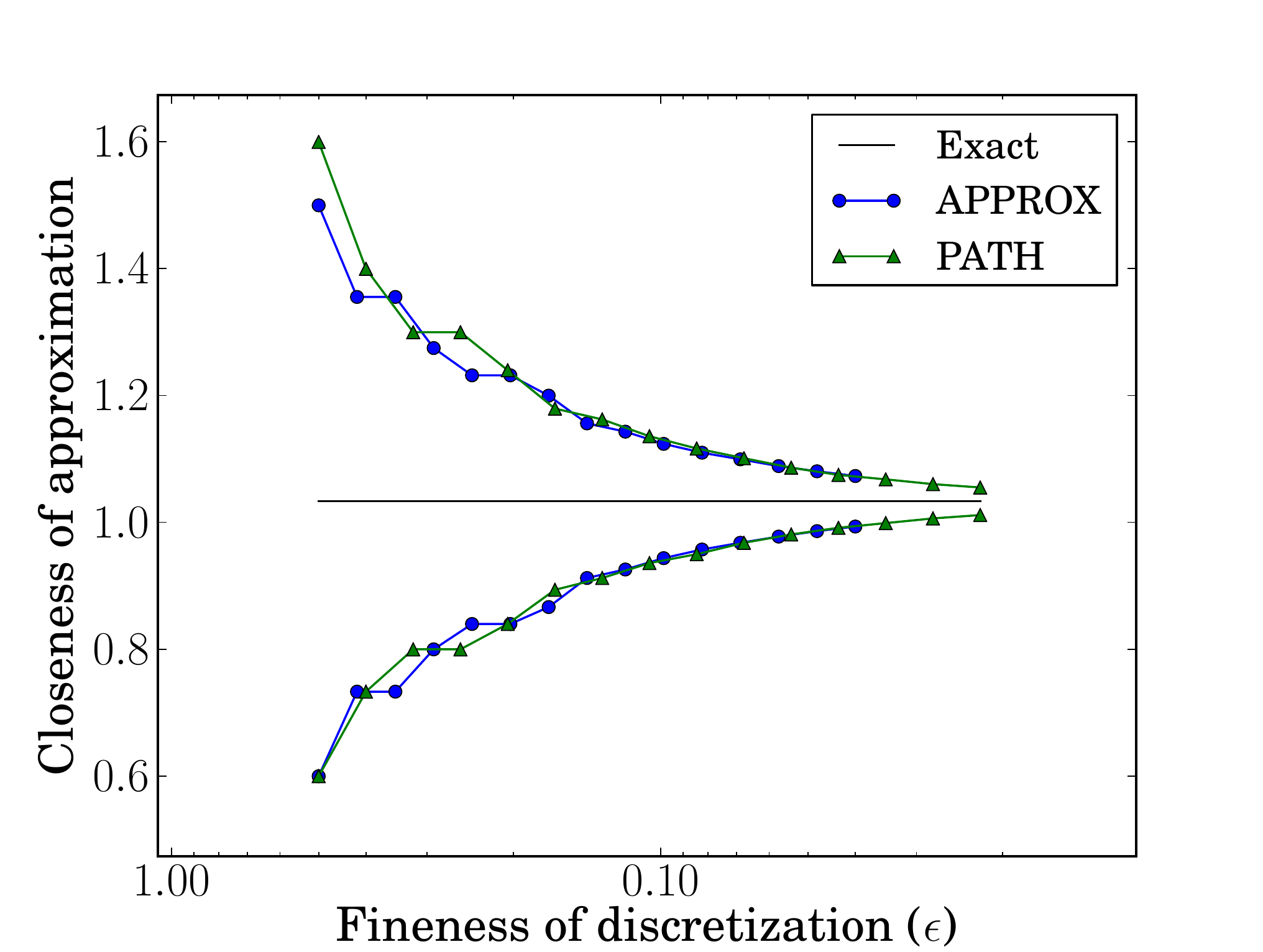}
\label{fig:numerical closeness}
}
\hfill
\subfigure[
Runtime of the network flow-based algorithm.
]{
\includegraphics[width=.45\linewidth]{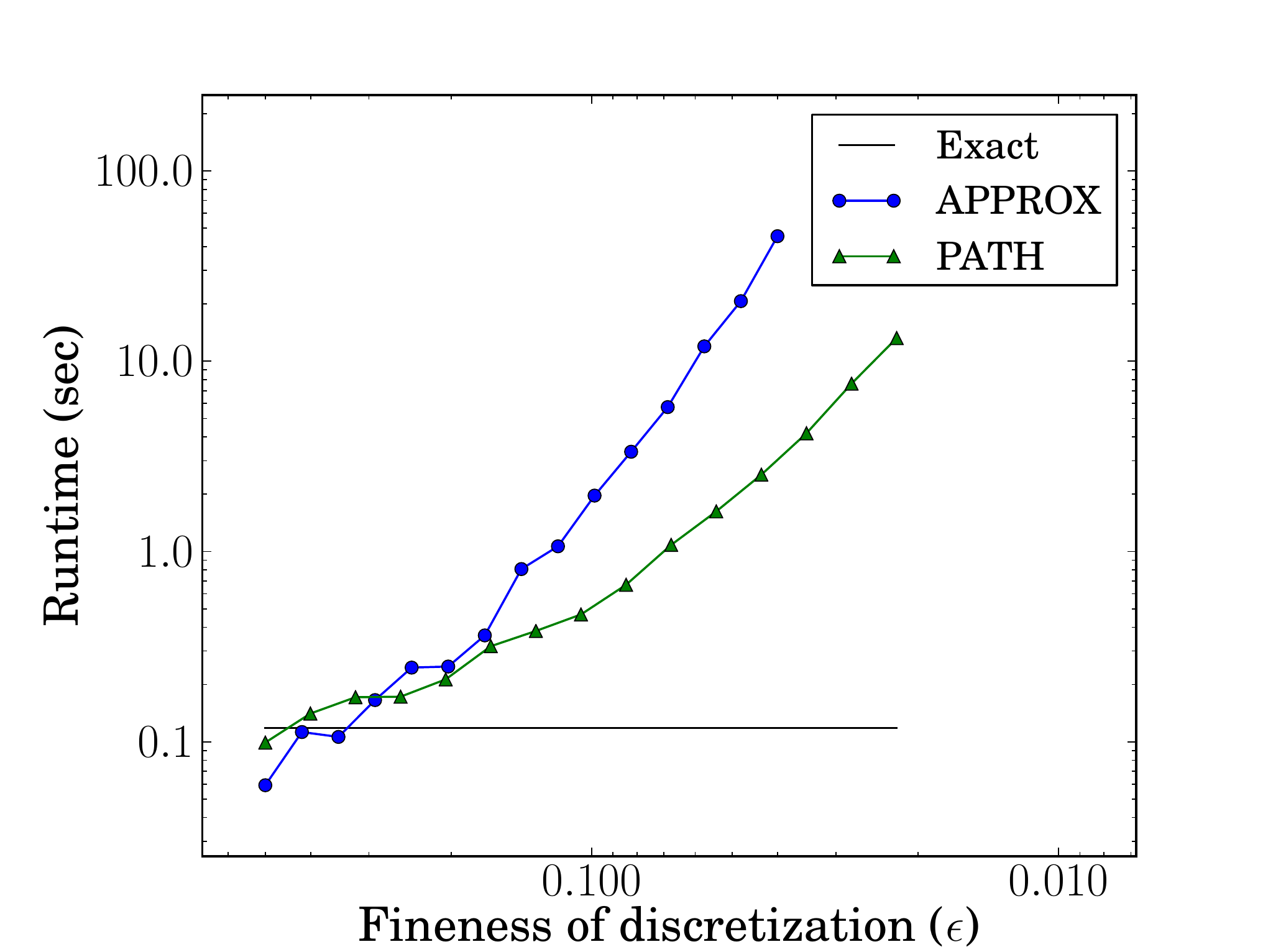}
\label{fig:numerical runtime}
}
\hfill\null
\caption{
Quality of approximation and runtime
as a function of the fine-ness parameter $\epsilon$.
The flat lines indicate value achieved by $\wassnet$, which is independent of $\epsilon$.
}
\end{figure}

\section{Evaluating the Limit of the Path-based Approximation}
\label{sec:proof of correctness}

$\pathnet$ is sufficiently structured that it will allow us to calculate the limit of~\eqref{eq:defn underapprox}
as $\epsilon \to 0^+$.
As argued in Section~\ref{sec:gp approx bounds}, that limit is equal to the EMD between the argument distributions.
In this section we present a derivation of the limit, which
produces the formulation of $\wassnet$ in Section~\ref{sec:main result construction}.

Suppose we are trying to compute the EMD between $\measone$ and $\meastwo$ over a road network $\roadnet$.
Let $\wassnet$ denote the resulting $\wasstag$ Wasserstein network, with edge costs $\wasscosts$;
let $\pathnet$ be the $\pathtag$ network generated by some $\epsilon$-tesselation of $\roadnet$,
with weights $\pathwts$.
Note that
the routing edges $\routingedges$ are present in both networks, so
the two networks differ only between
the decision edges $\decisionedges$ in $\wassnet$ and
the road devices in $\pathnet$.


\subsection{Costs Associated with Road Devices}

Let $\flowmap^*$ be a minimum-cost flow on $\pathnet$, and
let us consider the cost associated with the device $\roaddevice$
of a non-transshipment road $\roadvar \in \supplyset \cup \demandset$.
As in Figure~\ref{fig:supply device}, let the vertices of $\roaddevice$ be ordered
($\roadnode[1],\roadnode[2],\ldots,\roadnode[\numcell]$)
from $\roadvar^-$ to $\roadvar^+$.

Suppose $\roadvar \in \supplyset$.
Then from inspection of the device in Figure~\ref{fig:supply device}, we can denote the cost associated with $\roaddevice$ by
\begin{equation}
\label{eq:roadcost_contribution}
	\flowcost_\roadvar( \flowmap^* ; \pathwts )
	\doteq
	\sum_{k=1}^{N-2}
    	\epsilon_\roadvar \ \flowmap^*( \roadnode[k], \roadnode[k+1] )
		+ \epsilon_\roadvar \ \flowmap^*( \roadnode[k+1], \roadnode[k] ).
\end{equation}
Let us call all the edges of the form
$(\roadnode[k],\roadnode[k+1])$ the \emph{forward} edges;
in a similar fashion,
we call all the edges of the form
$(\roadnode[k],\roadnode[k-1])$ the \emph{backward} edges;
here, we are letting $\roadnode[0]$ and $\roadnode[\numcell+1]$ denote symbolically
the vertices $\roadvar^-$ and $\roadvar^+$ (respectively).
%
%
Our ability to obtain a meaningful expression relies crucially on
an important technical property of minimum-cost flows on $\pathtag$ networks:

Note that between any adjacent vertices in $\roaddevice$,
positive flow can be supported only either on the forward edge or the backward edge;
otherwise, $\flowmap^*$ would be non-minimal by existence of a cycle.
We say that a vertex $\roadnode$ ``parts'' device $\roaddevice$
if all forward flows
(i.e., positive flows on forward edges)
are on one side of $\roadnode$ and
all backward flows are on the opposite side.
If such a \emph{parting vertex} exists, then we say the device is \emph{parted} [by the flow].

%
\begin{lemma}[Minimum-cost flows part road devices]
\label{lemma:threshold}
Let $\flowmap^*$ be a minimum-cost admissible flow on $\pathnet$,
generated by some $\epsilon$-tesselation of some road network $\roadnet$.
Then every road device in $\pathnet$ is parted by $\flowmap^*$.
\end{lemma}
%
%
\begin{proof}
The proof is by contradiction:
Assume that $\flowmap^*$ is a minimum-cost admissible flow,
but the device of some $\roadvar \in \supplyset$ is not parted.
(We give the proof only for $\roadvar \in \supplyset$, but
the proof for $\roadvar \in \demandset$ is symmetrical.)
Note that because $\roadvar \in \supplyset$,
then $\pathsupply(\roadnode[k]) \geq 0$ for $k=1,\ldots,\numcell$.
This implies that the backward flows are non-decreasing in magnitude from $\roadvar^+$ to $\roadvar^-$ and
and the forwards flows are non-decreasing from $\roadvar^-$ to $\roadvar^+$.
(Otherwise, $\flowmap^*$ would be either non-minimal, by existence of a positive-flow cycle,
or else not admissible, by violation of a flow conservation constraint.)
Since $\roaddevice$ is not parted by assumption, then the flow changes direction at least twice.
Thus, there are indices $k'$ and $k''$, $k' \leq k''$, such that
$\flowmap( \roadnode[k'], \roadnode[k'+1] ) > 0$ and
$\flowmap( \roadnode[k''+1], \roadnode[k''] ) > 0$.
In that case, the monotonicity of forward and backward flows implies the existence of a positive-flow cycle
somewhere between $k'$ and $k''$,
drawing a contradiction against optimality of $\flowmap^*$.
\end{proof}

The parting of the road devices is quite powerful, because
in combination with the flow conservation constraints~\eqref{eq:conservation},
it allows us to express the whole device cost~\eqref{eq:roadcost_contribution} in terms of the known supplies $\pathsupply$, and
ultimately, the density function $\den_\roadvar$.
%
\begin{lemma}[Costs of Parted Devices] \label{lemma:devicecost}
Let $\pathnet$ be the $\pathtag$ Wasserstein network for some $\epsilon$-tesselation
of a road network $\roadnet$
with argument distributions $\measone$ and $\meastwo$.
Let $\roadvar$ be some non-transshipment road and
let $\flowmap$ be any admissible flow on $\pathnet$ which parts $\roadvar$;
let $\thresh_\roadvar$ denote the index of the part of $\roaddevice$.
Then
\begin{equation} \label{eq:roadcost_outgoing}
  \flowcost_\roadvar(\flowmap;\pathwts)
  	= o(1)
  		+ \int_{y=0}^{ \thresh_\roadvar \times \epsilon_\roadvar} \den_\roadvar(y) \ y \ dy
		+ \int_{y=\thresh_\roadvar \times \epsilon_\roadvar}^{\roadlen_\roadvar} \den_\roadvar(y) \ [ \roadlen_\roadvar - y ] \ dy,
\end{equation}
\begin{align}
\flowmap( \roadconnleft )
	&= \cumden( \thresh_\roadvar \times \epsilon_\roadvar ; \den_\roadvar )
		+ o(1),
	\qquad \text{and} \label{eq:road flow left} \\
\flowmap( \roadconnright )
	&= \cumden( \roadlen_\roadvar - \thresh_\roadvar \times \epsilon_\roadvar ; \reverseden_\roadvar )
		+ o(1).
	\label{eq:road flow right}
\end{align}
\end{lemma}
The proof of the lemma is fairly technical, and is provided in Appendix~\ref{sec:device cost proof}.
%
\begin{lemma}[Costs of Parted Devices (Refined)] \label{lemma:devicecost readable}
Let $\roadvar$ be some non-transshipment road and
let $\flowmap$ be any admissible flow on $\pathnet$ which parts $\roadvar$.
Then
\begin{equation} \label{eq:devicecost readable}
  \flowcost_\roadvar(\flowmap;\pathwts)
  = \costform( \flowmap(\roadconnleft) ; \den_\roadvar )
  	+ \costform( \flowmap(\roadconnright) ; \reverseden_\roadvar )
  	+ o(1) .
\end{equation}
\end{lemma}
\begin{proof}
It is easy to show that
\[
	\int_0^y \den(y') \ y' \ dy' \equiv \costform( \cumden( y; \den) ; \den ).
\]
Thus, we can obtain the first term of~\eqref{eq:devicecost readable}
by combining the first integral of~\eqref{eq:roadcost_outgoing} with~\eqref{eq:road flow left}, and
saving off any low order terms
(recall that all $\costform$ are Lipschitz).
Similarly, we can obtain the second term of~\eqref{eq:devicecost readable},
by combining the second integral of~\eqref{eq:roadcost_outgoing} with~\eqref{eq:road flow right};
in that case, first, we put a change of variables $y' = \roadlen_\roadvar - y$ and a substitution by $\reverseden_\roadvar$.

\end{proof}

\subsection{Proving the Main Result}

Lemma~\ref{lemma:devicecost readable} provides the critical component
of the proof of the main result of the paper, i.e., Theorem~\ref{thm:roadnet emd equation}.
%
\begin{proof}[Proof of Theorem~\ref{thm:roadnet emd equation}]
We begin by proving that
$\min_{\flowmap \in \wassnet} \flowcost( \flowmap ; \wasscosts )
	\leq 
\Wass(\measone,\meastwo)$.
That proof is by showing that
\begin{equation} \label{eq:lowerbound strategy}
\min_{\flowmap \in \wassnet} \flowcost( \flowmap ; \wasscosts )
\leq o(1)
	+ \min_{\flowmap \in \pathnet} \flowcost( \flowmap ; \pathwts ),
\end{equation}
where $\pathnet$ is of the $\epsilon$-tesselation of $\roadnet$ for $\epsilon > 0$ arbitrarily small, so that
the lemma holds in the limit as $\epsilon \to 0^+$.
Let $\flowmap^*$ be a minimum-cost admissible flow on $\pathnet$, and
let $\flowmap$ be the network flow on $\wassnet$ defined by
\begin{align}
\flowmap( \roadconnleft )
	& := \flowmap^*( \roadconnleft )
			\label{eq:same tail flow}		\\
\flowmap( \roadconnright )
	& := \flowmap^*( \roadconnright )
		\qquad & \text{for all $\roadvar \in \supplyset \cup \demandset$, and}
			\label{eq:same head flow}		\\
\flowmap( \arcvar )
	& := \flowmap^*( \arcvar )
		\qquad & \text{for all $\arcvar \in \routingedges$}.
			\label{eq:same routing flow}
\end{align}
It is a simple exercise to show that $\flowmap$ is admissible, i.e., $\flowmap \in \wassnet$.
Applying Lemma~\ref{lemma:devicecost readable},
we observe that for every road $\roadvar \in \roadset$,
the difference between
the cost of the road device $\roaddevice$ in $\pathnet$ and
the combined cost of the decision edges $\roadconnleft$ and $\roadconnright$ in $\wassnet$
is $o(1)$.
The flows and weights on $\routingedges$ are identical in both networks, contributing no additional costs.
Therefore, the total difference in cost between $\flowmap$ and $\flowmap^*$ is $o(1)$.
By definition, the minimum-cost flow on $\wassnet$ has cost bounded by $\flowcost( \flowmap ; \wasscosts )$,
and so we obtain~\eqref{eq:lowerbound strategy}.

We prove the matching lower bound by another limiting expression
\begin{equation}
\min_{\flowmap \in \pathnet} \flowcost( \flowmap ; \pathwts )
\leq o(1)
	+ \min_{\flowmap \in \wassnet} \flowcost( \flowmap ; \wasscosts ),
\end{equation}
Let $\flowmap^*$ be a minimum-cost admissible flow on the flow network $\wassnet$.
$\flowmap$ shall be an admissible flow ($\flowmap \in \pathnet$) satisfying
again~\eqref{eq:same tail flow}, \eqref{eq:same head flow}, and~\eqref{eq:same routing flow}.
$\flowmap$ must also part every device $\roaddevice$.
(Such $\flowmap$ can be generated, e.g., by traversing each device $\roaddevice$ and assigning flows greedily
to obtain~\eqref{eq:same tail flow} and~\eqref{eq:same head flow}.)
The rest of the proof continues by symmetrical logic.
\end{proof}

\section{Conclusion}
\label{sec:conclusion}

In this paper we have defined the Earth Mover's distance with respect to a set of
ground metrics capturing the common notion of ``roadmap distance''.
In order to produce such ground metrics,
we have defined formally a class of one-dimensional metric spaces which are
$\reals^1$-like but may have arbitrary, graph-like topology.
We have given an expression of the EMD on such road networks,
for a general class of probability distributions, which is explicit in the sense that
it is amenable to efficient computational optimization techniques.
In the case that both distributions are piece-wise uniform, the EMD can be computed by quadratic programming.
Finally, we have demonstrated by simulation experiment
that our formulation
can be used to predict accurately the maximum theoretical throughput of a vehicle sharing system 
modeled by the DPDP in a roadmap workspace.
The result can be used to address a limitation of previous DPDP models, which
treat the distances between points in a planar workspace using a simplified Euclidean distance metric.

\emph{Future Work:}
There are several directions is which this work may be extended.
For example, the authors are quite certain that the basic formulation shall admit simple extensions for
(i) the class of \emph{mixed} distributions,
i.e., distributions having an absolutely continuous part and an \emph{atomic} part;
(ii) non-symmetrical ground metrics resulting from the treatment of ``one-way'' streets.
It should also be straightforward to obtain a generalization of the formulation
for definitions of the EMD (e.g., in~\cite{rubner1998metric})
which allow input measures to have unequal total mass.
Another possible extension of this work would be to obtain better algorithms for road networks with special structure.
(For example, it should be possible to produce an algorithm in the style of~\cite[Sec.~5.3]{ling2007efficient}
for road networks that can be represented by tree graphs.)

In addition to these particular extensions, we hope that
our formal treatment of road networks and
the analytical techniques introduced in this paper
may facilitate bringing the power of computational statistics research to bear
on research questions framed in the ubiquitous road network setting.

\bibliographystyle{plain}
\bibliography{main}

%

\appendices

\section{Correctness of the General Purpose Approximation}
\label{sec:naive approx correctness}

Before proving the two propositions, we must introduce a relation
between the set of couplings $\Gamma(\measone,\meastwo)$
and the network flow constraints on $\approxnet$.

\begin{lemma}[Coupling-induced network flow]
\label{lemma:couplingflow}
Let $\measone$ and $\meastwo$ be two measures over a domain $\env$.
Let $\cellset$ be a partition of $\env$ into cells,
and let $\approxnet$ be the approximation network derived from $\measone$, $\meastwo$, and $\cellset$.
Let $\gamma$ be a coupling of measures $\measone$ and $\meastwo$,
$\gamma \in \Gamma(\measone,\meastwo)$.
Let $\flowmap : \tverts \times \bverts \to \reals$ be the mapping where 
%
\begin{equation} \label{eq:couplingflow_network}
  \flowmap( u, u' ) =
	\gamma\left( \workcell \times \workcell' \right)
  \qquad
  \text{for each $(u,\workcell) \in \tmatch$ and $(u',\workcell') \in \bmatch$.}
\end{equation}
Then $\flowmap$ is admissible, i.e., $\flowmap \in \approxnet$.
\end{lemma}
\begin{proof}
To prove $\flowmap$ of~\ref{eq:couplingflow_network} is admissible,
one must show that~\eqref{eq:conservation} holds.
On the bipartite network $\approxnet$,
\eqref{eq:conservation} holds if
$\sum_{ (u',\workcell') \in \bmatch } \flowmap( u, u' ) = \measone(\workcell)$
for all $(u,\workcell) \in \tmatch$
and
$\sum_{ (u,\workcell) \in \tmatch } \flowmap( u, u' ) = \meastwo(\workcell')$
for all $(u',\workcell') \in \bmatch$.
Recalling that $\gamma \in \Gamma(\measone,\meastwo)$, these conditions can be easily verified.
%
\end{proof}

\begin{prop}
\label{remark:corresp_exists}
For any admissible flow $\flowmap \in \approxnet$,
there exists at least one coupling $\gamma \in \Gamma(\measone,\meastwo)$
satisfying~\eqref{eq:couplingflow_network}.
(In general, there are many.)
\end{prop}
\begin{proof}
The proof is by an example construction.
Given $\flowmap \in \approxnet$, let $\gamma$ be the unique measure satisfying
\[
	\gamma( A \times B )
	=
	\sum_{ (u,\workcell) \in \tmatch, (u',\workcell') \in \bmatch }
		\flowmap(u,u') \,
		\frac{ \measone( A \cap \workcell) }{ \measone(\workcell) } \,
		\frac{ \meastwo( B \cap \workcell') }{ \meastwo(\workcell') }
\]
for all $A, B \in \sigfield$
(with the standard extension to the product measure-space $\sigfield \otimes \sigfield$).
It can be checked that $\gamma$ satisfies the conditions of the proposition.
\end{proof}

\begin{proof}[Proof of Prop.~\ref{prop:approx_contains_wass}]
First, we show that $\underapprox(\epsilon) \leq \Wass$ for all $\epsilon > 0$;
For the rest of the proof, we will omit the argument $\epsilon$.
For $\delta > 0$ arbitrarily small,
we choose some $\gamma \in \Gamma(\measone,\meastwo)$
within $\delta$ of the infimum~\eqref{eq:EMDdef}.
Let $\flowmap$ be given by~\eqref{eq:couplingflow_network}.
%
Then we have
%
\begin{equation}
\begin{aligned}
  \Wass
  =
  \inf_{ \gamma' } \int \distfunc{\pointvar,\pointvar'} d\gamma'(\pointvar,\pointvar')
  \geq
  \int \distfunc{\pointvar,\pointvar'} d\gamma(\pointvar,\pointvar') + \delta.
\end{aligned}
\end{equation}
Let us define the distance function
\begin{equation}
\underdist_\cellset(\pointvar,\pointvar')
  :=
  \sum_{\workcell,\workcell' \in \cellset}
  \krondelta_{ \{ \pointvar\in\workcell, \pointvar'\in\workcell' \} }
  \min_{ \altpointvar \in \workcell, \altpointvar' \in \workcell' }
  \distance(\altpointvar,\altpointvar').
\end{equation}
We observe $\underdist_\cellset$ is everywhere a lower bound for $\distance$;
therefore,
\begin{equation}
\int \distfunc{\pointvar,\pointvar'} d\gamma(\pointvar,\pointvar')
\geq
\int \underdist_\cellset(\pointvar,\pointvar') d\gamma(\pointvar,\pointvar').
\end{equation}
Letting $\lowerwts =: \{ \lowerwt_\arcvar \}_{\arcvar \in \arcset}$, note that
\begin{equation}
\begin{aligned}
  \int \underdist_\cellset(\pointvar,\pointvar') d\gamma(\pointvar,\pointvar')
  &=
  \sum_{ \workcell,\workcell' \in \cellset  }
	\min_{ \altpointvar \in \workcell, \altpointvar \in \workcell' }
    \distfunc{\altpointvar,\altpointvar'}
    \
    \int_{ \pointvar \in \workcell, \pointvar' \in \workcell' }
      d\gamma(\pointvar,\pointvar')	\\
  &=
	\sum_{ u \in \tverts, u' \in \bverts }
		\lowerwt_{(u,u')} \flowmap( u, u' )	\\
  &= \flowcost( \flowmap ; \lowerwts ).
\end{aligned}
\end{equation}
By definition, $\flowcost(\flowmap; \lowerwts )$ is no smaller than $\underapprox$ .
Combining these results we have that $\Wass \geq \underapprox + \delta$.
The proof follows since the inequality holds for $\delta$ arbitrarily small.

The proof that $\overapprox \geq \Wass$ is similar.
Let $\flowmap$ be the minimum-cost flow of $\approxnet$ under edge weights $\upperwts$;
by definition, the cost of $\flowmap$ is $\overapprox$.
Recalling Remark~\ref{remark:corresp_exists},
let $\gamma$ be any coupling of $\measone$ and $\meastwo$ which induces $\flowmap$.
Then
\begin{equation}
\begin{aligned}
  \Wass
  = \inf_{\gamma'} \int \distfunc{\pointvar,\pointvar'} \ d\gamma'(\pointvar,\pointvar')
  \leq \int \distfunc{\pointvar,\pointvar'} \ d\gamma(\pointvar,\pointvar').
\end{aligned}
\end{equation}
We define the distance function
\begin{equation}
\overdist_\cellset(\pointvar,\pointvar')
  :=
  \sum_{\workcell,\workcell' \in \cellset}
  \krondelta_{ \{ \pointvar\in\workcell, \pointvar'\in\workcell' \} }
  \max_{ \altpointvar \in \workcell, \altpointvar' \in \workcell' }
  \distance(\altpointvar,\altpointvar');
\end{equation}
$\overdist_\cellset$ is everywhere greater than $\distance$, so
\begin{equation}
\begin{aligned}
  \int \distfunc{\pointvar,\pointvar'} d\gamma(\pointvar,\pointvar')
  \leq
  \int \overdist_\cellset(\pointvar,\pointvar') d\gamma(\pointvar,\pointvar').
\end{aligned}
\end{equation}
By previous logic, it can be shown that
\[
\int \overdist_\cellset(\pointvar,\pointvar') \ d\gamma(\pointvar,\pointvar')
  = \flowcost( \flowmap ; \upperwts ) = \overapprox(\epsilon). 
\]
Combining these results proves the second part.
\end{proof}

\begin{proof}[Proof of Prop.~\ref{prop:approx_squeezes}]
The result is simply a consequence of the fact (one can check) that
for any $\epsilon > 0$, and
$\lowerwts(\epsilon) =: \{ \lowerwt_\arcvar \}_{\arcvar \in \arcset}$,
$\upperwts(\epsilon) =: \{ \upperwt_\arcvar \}_{\arcvar \in \arcset}$, we have
$\upperwt_\arcvar - \lowerwt_\arcvar \leq \epsilon$ for all $\arcvar \in \arcset$.
Let $\flowmap^*$ be the minimum-cost flow on $\approxnet$
with edge weights $\lowerwts$.
Note that 
\begin{equation}
\begin{aligned}
  \overapprox(\epsilon)
  &=
  \min_{ \newflowmap \in \approxnet } \flowcost( \newflowmap ; \upperwts(\epsilon) )	\\
  &\leq
  \flowcost( \flowmap^* ; \upperwts(\epsilon) )		\\
  &=
  \sum_{u \in \tverts, u' \in \bverts }
    \upperwt_{ (u,u') } \flowmap^*( u,u' )		\\
  &\leq
  \sum_{u \in \tverts, u' \in \bverts }
    \left[ \lowerwt_{ (u,u') } + \epsilon \right]
    \flowmap^*( u, u' )		\\
  &=
  \flowcost( \flowmap^*; \lowerwts(\epsilon) ) + \epsilon |\mu|	
  =
  \underapprox(\epsilon) + \epsilon |\mu|.	
\end{aligned}
\end{equation}
%
%
\end{proof}

\section{Reimann Approximation of Road Device Costs}
\label{sec:device cost proof}

\newcommand{\roadcostLower}{\flowcost_\roadvar^-}
\newcommand{\roadcostUpper}{\flowcost_\roadvar^+}

\begin{proof}[Proof of Lemma~\ref{lemma:devicecost}]
We give the proof only for $\roadvar \in \supplyset$; the proof for $\roadvar \in \demandset$ is by identical logic.
%
Since $\roaddevice$ is parted,
we can restrict the ranges of the sums in~\eqref{eq:roadcost_contribution} to obtain
\begin{equation}
\label{eq:roadcost_split}
  \flowcost_\roadvar( \optflow ; \pathwts )
  =
  \sum_{k=0}^{ k_\roadvar - 1 } \epsilon_\roadvar \ \optflow( \roadnode[k+1], \roadnode[k] )
  + \sum_{k=k_\roadvar}^{ \numcell - 1 } \epsilon_\roadvar \ \optflow( \roadnode[k], \roadnode[k+1] ).
\end{equation}
%
Combining the parted-ness of $\roaddevice \in \supplyset$
with the flow conservation constraints~\eqref{eq:conservation},
we obtain a recursive system
%
\begin{align}
\optflow( \roadnode[k], \roadnode[k-1] )
	& = \optflow( \roadnode[k+1], \roadnode ) + \supplyvar( \roadnode ),
	& \text{for $k=2,\ldots, \thresh_\roadvar-1$, and}
	\label{eq:flowconserve_left}		\\
\optflow( \roadnode, \roadnode[k+1] ) 
	& = \optflow( \roadnode[k-1], \roadnode ) + \supplyvar( \roadnode ),
	& \text{for $k=\thresh_\roadvar+1,N-2$}.
	\label{eq:flowconserve_right}
\end{align}
We can ``unroll'' each of the recursions~\eqref{eq:flowconserve_left} and~\eqref{eq:flowconserve_right}
until we reach the part index $\thresh_\roadvar$;
since the supply $\supplyvar( k_\roadvar )$ could be split between the backward and forward flows,
at best we can write bounds
\begin{align} 
  \sum_{k'=k+1}^{ \thresh_\roadvar-1 } \supplyvar( \roadnode[k'] )
  \leq
  \optflow( \roadnode[k+1], \roadnode )
  \leq
  \sum_{k'=k+1}^{ \thresh_\roadvar } \supplyvar( \roadnode[k'] )
 & & \text{for all $k < \thresh_\roadvar$},
	\label{eq:cumulativeflow_left}		\\
  \sum_{ k'=\thresh_\roadvar+1 }^k \supplyvar( \roadnode[k'] )
  \leq
  \optflow( \roadnode, \roadnode[k+1] )
  \leq
  \sum_{ k'=\thresh_\roadvar }^k \supplyvar( \roadnode[k'] )
  && \text{for all $k \geq \thresh_\roadvar$}.
\label{eq:cumulativeflow_right}
\end{align}
Substituting~\eqref{eq:cumulativeflow_left} and~\eqref{eq:cumulativeflow_right}
in~\eqref{eq:roadcost_split},
and re-arranging the sums,
we obtain bounds
$\roadcostLower \leq \roadcost \leq \roadcostUpper$, where
\begin{align} 
	\roadcostLower &:=
	\sum_{k'=1}^{ \thresh_\roadvar - 1 } \supplyvar( \roadnode[k'] ) \ k' \epsilon_\roadvar
    +
	\sum_{k'= \thresh_\roadvar + 1 }^{\numcell-1} \supplyvar( \roadnode[k'] ) \ (\numcell-k') \ \epsilon_\roadvar,
\label{eq:roadcost_lowerbound} 	\\
	\roadcostUpper & :=
	\sum_{k'=1}^{ \thresh_\roadvar } \supplyvar( \roadnode[k'] ) \ k'\epsilon_\roadvar
    +
	\sum_{ k' = \thresh_\roadvar }^{ \numcell-1 } \supplyvar( \roadnode[k'] ) \ (\numcell-k') \ \epsilon_\roadvar.
	\notag
\end{align}
%
The two bounds have separation
$\roadcostUpper - \roadcostLower
	= \supplyvar( \roadnode[\thresh_\roadvar] ) \ [ \thresh_\roadvar + (\numcell-\thresh_\roadvar) ] \ \epsilon_\roadvar$
  	$= \supplyvar( \roadnode[\thresh_\roadvar] ) \ \roadlen_\roadvar$.
Since $\den_\roadvar$ is Lipschitz by assumption, then
\begin{equation} \label{eq:flow_into_cell}
  \supplyvar( \roadnode )
  = \epsilon_\roadvar \ \den_\roadvar( k\epsilon_\roadvar ) + o(\epsilon_\roadvar)
  \qquad
  \text{for $k=1,\ldots,\numcell$},
\end{equation}
and so $\roadcost = \roadcostLower + O(\epsilon)$.
Substituting~\eqref{eq:flow_into_cell} into~\eqref{eq:roadcost_lowerbound},
as well as
$y(k) \doteq k\epsilon_\roadvar$ and $\Delta y \doteq \epsilon_\roadvar$,
we obtain
\begin{equation} \label{eq:lowerbound_reimann_sum}
\roadcost =
O(\epsilon)
	+ \sum_{k'=1}^{ \thresh_\roadvar - 1 }
		\left[ \den_\roadvar(y(k')) \Delta y + o(\Delta y) \right] \ y(k')
  	+ \sum_{k'=\thresh_\roadvar+1}^{N-1}
    	\left[ \den_\roadvar(y(k')) \Delta y + o(\Delta y) \right] \ ( \roadlen_\roadvar - y(k') );
\end{equation}
\eqref{eq:lowerbound_reimann_sum} is a Reimann sum which can be written as~\eqref{eq:roadcost_outgoing}.

\eqref{eq:road flow left} and~\eqref{eq:road flow right}
can be obtained in a similar fashion
by substituting~\eqref{eq:flow_into_cell} into
\eqref{eq:cumulativeflow_left} and~\eqref{eq:cumulativeflow_right},
for
$\flowmap( \roadconnleft ) \equiv \flowmap( \roadnode[1], \roadnode[0] )$ and
$\flowmap( \roadconnright ) \equiv \flowmap( \roadnode[N], \roadnode[N+1] )$, respectively, then
identifying the Reimann sums, and
applying Definition~\ref{def:cumden}.
To obtain~\eqref{eq:road flow right} also requires a change of variables $y' = \roadlen_\roadvar - y$ and
substitution by $\reverseden_\roadvar$.
\end{proof}

\end{document}